\documentclass[12pt]{amsart}
\usepackage{graphicx,amsmath,amssymb}

\numberwithin{equation}{section}

\newcommand{\C}{\mathbb C}
\newcommand{\R}{\mathbb R}
\newcommand{\Z}{\mathbb Z}
\newcommand{\N}{\mathbb N}
\renewcommand{\d}{\prime}
\newcommand{\dd}{{\prime \prime}}

\renewcommand{\Re}{{\rm Re}\,}

\newtheorem{theorem}{Theorem}[section]
\newtheorem{lemma}[theorem]{Lemma}
\newtheorem{proposition}[theorem]{Proposition}
\newtheorem{corollary}[theorem]{Corollary}

\newtheorem{definition}{Definition}

\newtheorem*{remarks}{Remarks}

\begin{document}
\title[]
{Anharmonic oscillators in the complex plane, $\mathcal{PT}$-symmetry, and real eigenvalues.}
\author[]
{Kwang C. Shin}
\address{Department of Mathematics, University of West Georgia, Carrollton, GA 30118 USA}
\date{\today}
\begin{abstract}
For integers $m\geq 3$ and $1\leq\ell\leq m-1$, we study the eigenvalue problems
$-u^\dd(z)+[(-1)^{\ell}(iz)^m-P(iz)]u(z)=\lambda u(z)$ with the boundary conditions that $u(z)$ decays to zero as
$z$ tends to infinity along the rays $\arg z=-\frac{\pi}{2}\pm \frac{(\ell+1)\pi}{m+2}$ in the complex plane,
where $P$ is a polynomial of degree at most $m-1$. We provide asymptotic expansions of the eigenvalues
$\lambda_{n}$. Then we show that if the eigenvalue problem is $\mathcal{PT}$-symmetric, then the eigenvalues are
all real and positive with at most finitely many exceptions. Moreover, we show that when $\gcd(m,\ell)=1$, the
eigenvalue problem has infinitely many real eigenvalues if and only if its translation or itself is
$\mathcal{PT}$-symmetric. Also, we will prove some other interesting direct and inverse spectral results.
\end{abstract}

\maketitle

\begin{center}
{\it 2010 {\it Mathematics Subject Classification}: 34L40, 34L20, 81Q12}\\
{\it  Key words: Anharmonic oscillators, asymptotics of the eigenvalues, $\mathcal{PT}$-symmetry}

\end{center}

\baselineskip = 18pt

\section{Introduction}
\label{introduction}
In this paper, we study Schr\"odinger eigenvalue problems with real and complex polynomial potentials in the complex plane under various decaying boundary conditions. We provide explicit asymptotic formulas relating the index $n$ to a series of fractional powers of the eigenvalue $\lambda_n$ (see Theorem \ref{main_thm1}). Also, we recover the polynomial potentials from asymptotic formula of the eigenvalues (see Theorem \ref{thm_112} and Corollary \ref{cor6}) as well as applications to the so-called $\mathcal{PT}$-symmetric Hamiltonians (see Theorems \ref{main_thm2} and \ref{thm_111}).

For integers $m\geq 3$  and $1\leq \ell\leq m-1$, we consider the Schr\"odinger eigenvalue
problem
\begin{equation}\label{ptsym}
\left(H_{\ell} u\right)(z):=\left[-\frac{d^2}{dz^2}+(-1)^{\ell}(iz)^m-P(iz)\right]u(z)=\lambda u(z),\quad\text{for
some $\lambda\in\C$},
\end{equation}
with the boundary condition that
\begin{equation}\label{bdcond}
\text{$u(z)\rightarrow 0$  as $z\rightarrow \infty$ along the two rays}\quad \arg z=-\frac{\pi}{2}\pm
\frac{(\ell+1)\pi}{m+2},
\end{equation}
where $P$ is a polynomial of degree at most $m-1$ of the form
\begin{equation}\nonumber
P(z)=a_1z^{m-1}+a_2z^{m-2}+\cdots+a_{m-1}z+a_m,\quad a_j\in\C\,\,\text{\,for $1\leq j\leq m$}.
\end{equation}

If a nonconstant function $u$ satisfies \eqref{ptsym} with some $\lambda\in\C$ and the boundary condition
\eqref{bdcond}, then we call $\lambda$ an {\it eigenvalue} of $H_{\ell}$ and $u$ an {\it eigenfunction of
$H_{\ell}$ associated with the eigenvalue $\lambda$}. Sibuya \cite{Sibuya} showed that the eigenvalues of $H_{\ell}$ are the zeros of an entire function of order
$\rho:=\frac{1}{2}+\frac{1}{m}\in (0,1)$ and hence, by the Hadamard factorization theorem (see, e.g., \cite[p.\ 291]{Conway}),
there are infinitely many eigenvalues. We call the entire function the Stokes multiplier (or the spectral
determinant), and the algebraic multiplicity of an eigenvalue $\lambda$ is the order of the zero $\lambda$ of the
Stokes multiplier. Also, the geometric multiplicity of an eigenvalue $\lambda$ is the number of linearly
independent eigenfunctions associated with the eigenvalue $\lambda$, that is $1$ for every eigenvalue $\lambda$ \cite[\S 7.4]{Hille}.

We number the eigenvalues $\{\lambda_{n}\}_{n\geq N_0}$ in the
order of nondecreasing magnitudes, counting their algebraic
multiplicities. We will show that the magnitude of large
eigenvalues is strictly increasing (see, Lemma \ref{monoton}) and
hence, there is a unique way of ordering large eigenvalues, but
this is not guaranteed for small eigenvalues. However, how we
order these small eigenvalues will not affect results in this
paper. Throughout this paper, we will use $\lambda_n$ to denote
the eigenvalues of $H_{\ell}$ without explicitly indicating their
dependence on the potential and the boundary condition. Also, we
let
$$a:=(a_1,a_2,\ldots, a_{m})\in\C^{m}$$ be the coefficient vector
of $P(z)$.

The anharmonic oscillators $H_{\ell}$ with the various boundary conditions \eqref{bdcond} are considered in
\cite{Bender2,Shin2}.  When $m$ is even and $\ell=\frac{m}{2}$, $H_{\frac{m}{2}}$ is a Schr\"odinger operator in
$L^2(\R)$ (see, e.g.,
\cite{AA,Bender-1,CGM1,Eremenko1,Eremenko2,HR,MAS,TIT2}). This is self-adjoint if the potential $V(z)=(-1)^{\ell}(iz)^m-P(iz)$ is real on the real line, and
non-self-adjoint if the potential is non-real.

 Some particular classes of $H_{1}$ have been studied extensively in recent years in the context of theory of
$\mathcal{PT}$-symmetry \cite{Bender,CGM,Dorey,Handy2,Ali1,Shin5,Znojil,Znojil2}. The $H_{\ell}$ is
$\mathcal{PT}$-symmetric if the potential $V$ satisfies $\overline{V(-\overline{z})}=V(z)$, $z\in\C$, that is equivalent to $a\in\R^{m}$. In  this paper, we will generalize results in \cite{Shin5} (where $H_1$ is studied) to $1\leq \ell\leq m-1$ and introduce some new results.
These results are consequences of the following asymptotic expansion of the eigenvalues.
\begin{theorem}\label{main_thm1}
For each integer $m\geq 3$ and $1\leq \ell\leq m-1$, there exists an integer $N_0=N_0(m,\ell)$ such that the
eigenvalues $\{\lambda_n\}_{n\geq N_0}$ of $H_{\ell}$ satisfy
\begin{equation}\label{main_result}
n+\frac{1}{2}
=\sum_{j=0}^{m+1}c_{\ell,j}(a)\lambda_{n}^{\frac{1}{2}+\frac{1-j}{m}}+\eta_{\ell}(a)+O\left(\lambda_n^{-\rho}\right)
\quad\text{as $n\to\infty$},
\end{equation}
where $c_{\ell,j}(a)$ and $\eta_{\ell}(a)$ are defined in \eqref{dmlj_def} and \eqref{eta_def}, respectively.
\end{theorem}
Also, we obtain the partial reality of the eigenvalues for $\mathcal{PT}$-symmetric $H_{\ell}$.
\begin{theorem}\label{main_thm2}
If $H_{\ell}$ is $\mathcal{PT}$-symmetric, then eigenvalues are all real with at most finitely many exceptions.
\end{theorem}
\begin{proof}
If $u(z)$ is an eigenfunction associated with the eigenvalue $\lambda$ of a $\mathcal{PT}$-symmetric $H_{\ell}$,
then $\overline{u(-\overline{z})}$ is also an eigenfunction associated with $\overline{\lambda}$. In Corollary
\ref{monoton}, we will show that $|\lambda_n|<|\lambda_{n+1}|$ for all large  $n$ and $\arg(\lambda_n)\to 0$.
Thus, $\lambda_n$ are  real and positive for all large $n$ since $|\lambda|=|\overline{\lambda}|$.
\end{proof}
Many $\mathcal{PT}$-symmetric operators have  real eigenvalues only \cite{Dorey,Shin,Shin4}. However, there are some $\mathcal{PT}$-symmetric $H_{\ell}$ that produce a finite number of non-real eigenvalues \cite{Bender-1,Delabaere,Dorey1,Handy2}. When $H_{\ell}$ is self-adjoint, the spectrum is real. Conversely, when the spectrum is real, what can we conclude about $H_{\ell}$?  The next theorem provides a necessary and sufficient condition for $H_{\ell}$ to have infinitely many real eigenvalues.
\begin{theorem}\label{thm_111}
Suppose that $\gcd(m,\ell)=1$. Then $H_{\ell}$ with the potential $V(z)=-(iz)^m+P(iz)$ has infinitely many real
eigenvalues if and only if $H_{\ell}$ with the potential $V(z-z_0)$ for some $z_0\in\C$ is
$\mathcal{PT}$-symmetric.
\end{theorem}

The next theorem reveals an interesting feature of the
eigenvalues as a sequence:
\begin{theorem}\label{thm_112}
Suppose that $\gcd(m,\ell)=1$. Let $\{\lambda_n\}_{n\geq N_0}$ and
$\{\widetilde{\lambda}_n\}_{n\geq N_0}$ be the eigenvalues of
$H_{\ell}$  with the potentials $V$ and $\widetilde{V}$,
respectively. Suppose that
$\lambda_n-\widetilde{\lambda}_n=o\left(1\right)$ as $n\to\infty$.
Then  $\widetilde{V}(z)=V(z-z_0)$ for some $z_0\in\C$ and $\lambda_n=\widetilde{\lambda}_n$ for all $n\geq N_0$ after, if needed, small eigenvalues are reordered.
\end{theorem}

The asymptotic expansions of the eigenvalues of $H_{\ell}$ with $\ell=\lfloor\frac{m}{2}\rfloor$ have been studied
in, for example, \cite{Fedoryuk,HR,MAS}.
 Maslov \cite{MAS}  computed the first three terms of asymptotic expansions of $\lambda_n^{\frac{3}{4}}$,
 where $\lambda_n$ are the eigenvalues of $-\frac{d^{2}}{dx^{2}}u+x^{4}u=\lambda u,\,\, u\in L^2(\R).$
 Helffer and Robert \cite{HR} considered
\begin{equation}\nonumber
-\frac{d^{2k}}{dx^{2k}}u+(x^{2m}+p(x))u=\lambda u,\quad u\in
L^2(\R),
\end{equation}
where $k,\, m$ are positive integers and where $p(\cdot)$ is a {\it real}  polynomial of degree at most $2m-1$.
They obtained the existence of asymptotic expansions of the eigenvalues to all orders, and suggested an explicit
way of computing the coefficients of the asymptotic expansion. In particular, for the case when the potential is
$\varepsilon x^4+x^2$, $\varepsilon>0$, Helffer and Robert \cite{HR} computed the first nine terms of the
asymptotic expansion of $\lambda_n^{\frac{3}{4}}$. Also, Fedoryuk \cite[\S 3.3]{Fedoryuk} considered \eqref{ptsym}
with complex polynomial potentials and with \eqref{bdcond} for $\ell=\lfloor\frac{m}{2}\rfloor$ and computed the
first term in the asymptotic expansion. Also, Sibuya \cite{Sibuya} computed the first term in the asymptotic
expansion for $\ell=1$.

This paper is organized as follows. In Section \ref{sec_cor}, we define $c_{\ell,j}(a)$ and  some other notations. Also, we invert \eqref{main_result}, expressing $\lambda_n$ as a series of fractional powers of the index $n$ and  prove Theorems \ref{thm_111} and \ref{thm_112}, and other interesting direct and inverse spectral results.  In Section \ref{prop_sect}, we introduce some properties of the solutions of the differential equation in \eqref{ptsym}, due to  Hille \cite{Hille} and Sibuya \cite{Sibuya}. We study the asymptotic of the Stokes multiplier associated with $H_1$ in Section \ref{sec_4} and treat the general case $H_{\ell}$ in Section \ref{sec_5}. In Section \ref{asymp_eigen}, we relate the eigenvalues of $H_{\ell}$ with the zeros of the Stokes multiplier. We prove Theorem \ref{main_thm1} for $1\leq \ell<\frac{m}{2}$ in Section \ref{sec_7} and for $\ell\geq\frac{m}{2}$ in Section \ref{sec_8}.

\section{Notations and corollaries}\label{sec_cor}
In this section, we will define $c_{\ell,j}(a)$ and some other notations and introduce some corollaries of Theorem \ref{main_thm1}.

We define, for  nonnegative integers $k, j$,
\begin{equation}
\text{$b_{j,k}( {a})$ is the coefficient of $\frac{1}{z^{j}}$ in
${\frac{1}{2}\choose{k}}\left(\frac{a_1}{z}+\frac{a_2}{z^2}+\dots+\frac{a_m}{z^m}\right)^k$\,\,\, and}
\end{equation}
\begin{equation}
b_j( {a})=\sum_{k=0}^j b_{j,k}( {a}),\quad j\in\N.
\end{equation}
Notice that $b_{0,0}(a)=b_0(a)=1$, $b_{j,0}(a)=0$ if $j\geq 1$, and $b_{j,k}(a)=0$ if $j<k$ or $k<\frac{j}{m}$.
We will also use
\begin{equation}\label{Kmj=}
K_{m,j}( {a})=\sum_{k=0}^j K_{m,j,k}\,b_{j,k}( {a}) \,\,\, \text{for $j\geq 0$,}
\end{equation}
where for $k<\frac{j}{m}$, $K_{m,j,k}=0$ and for $\frac{j}{m}\leq k\leq j$,
\begin{align}\nonumber
K_{m,j,k}&:=\left\{
                    \begin{array}{cl}
                    \frac{B\left(\frac{1}{2},1+\frac{1}{m}\right)}{2\cos\left(\frac{\pi}{m}\right)}
\quad &\text{if $j=k=0$},\\
&\\
-\frac{2}{m}
\quad &\text{if $j=k=1$},\\
&\\
\frac{2}{m}\left(\ln 2-\sum_{s=1}^{k-1}\frac{1}{2s-1}\right) \, &\text{if $j=\frac{m}{2}+1$, $m$ even,}\\
&\\
                  \frac{1}{m}B\left(k-\frac{j-1}{m},\,\frac{j-1}{m}-\frac{1}{2}\right) \, &
                  \text{if  $j\not=1$ or $j\not=\frac{m}{2}+1$.}
\end{array}\right.\nonumber
\end{align}
Now we are ready to define $c_{\ell,j}(a)$ as follows: for $1\leq j\leq m+1$,
\begin{equation}\label{dmlj_def}
c_{\ell,j}(a)=-\frac{2}{\pi}\sum_{k=0}^j(-1)^{(\ell+1)k}K_{m,j,k}b_{j,k}(a)\sin\left(\frac{(j-1)\ell\pi}{m}\right)\cos\left(\frac{(j-1)\pi}{m}\right),
\end{equation}
and we also define
\begin{equation}\label{eta_def}
\eta_{\ell}(a)=\left\{
              \begin{array}{rl}
             (-1)^{\frac{\ell-1}{2}}\frac{2\nu(a)}{m}\quad &\text{if $\ell$ is odd,}\\
             &\\
              0 \quad &\text{if otherwise}
              \end{array}
                         \right. \,\, \text{and }\,\, \mu( {a})=\frac{m}{4}-\nu( {a}),
\end{equation}
where
\begin{eqnarray}
\nu( {a})=\left\{
              \begin{array}{rl}
              b_{\frac{m}{2}+1}( {a}) \quad &\text{if $m$ is even,}\\
              &\\
              0 \quad &\text{if $m$ is odd.}
              \end{array}
                         \right. \nonumber
\end{eqnarray}

\subsection{Further direct spectral results}
One can invert the asymptotic formulas \eqref{main_result} to obtain formulas for $\lambda_n$ in terms of $n$.
\begin{corollary}
 One can compute $d_{\ell,j}( {a})$ explicitly such that
\begin{equation}\label{somm_eq}
\lambda_n=\sum_{j=0}^{m+1}d_{\ell,j}(
{a})\left(n+\frac{1}{2}\right)^{\frac{2m}{m+2}\left(1-\frac{j}{m}\right)}+O\left(n^{-\frac{4}{m+2}}\right),
\end{equation}
where $d_{\ell,0}(a)=
\left(\pi^{-1}B\left(\frac{1}{2},1+\frac{1}{m}\right)\sin\left(\frac{\ell\pi}{m}\right)\right)^{-\frac{2m}{m+2}}>0$.
\end{corollary}
\begin{proof}
Equations \eqref{main_result} is an asymptotic equation and it can be solved for $\lambda_n$, resulting in
\eqref{somm_eq}. For details,  see, for example, \cite[\S 5]{Shin2}.
\end{proof}

In the next corollary, we provide an asymptotic formula for the nearest neighbor spacing of the eigenvalues. And
the large eigenvalues increase monotonically in magnitude and have their argument approaching zero.
\begin{corollary}\label{monoton}
The space between successive eigenvalues is
\begin{equation}
 \lambda_{n+1}-\lambda_n
\underset{n\to+\infty}{=}\frac{2m}{m+2}\,d_{\ell,0} \cdot
\left(n+\frac{1}{2}\right)^{\frac{m-2}{m+2}}+o\left(n^{\frac{m-2}{m+2}}\right).
\end{equation}
In particular, $\lim_{n\to\infty}\left|\lambda_{n+1}-\lambda_n\right|=\infty$ and
$\lim_{n\to+\infty}\arg(\lambda_n)=0.$ Hence:

\begin{equation}\nonumber
\left|\lambda_n\right|<\left|\lambda_{n+1}\right|\,\,\,\text{for all large $n$}.
\end{equation}
\end{corollary}
\begin{proof}
 These claims are consequences of \eqref{somm_eq} and the generalized binomial expansion. For details,  see, for example, \cite[\S 6]{Shin2}.
\end{proof}

Suppose that $\lambda_n$ are eigenvalues of $H_{\ell}$ for some  $P$. Then the degree $m$ of the polynomial
potential can be recovered by
\begin{equation}\nonumber
\frac{2m}{m+2}=\lim_{n\to\infty}\frac{\ln(|\lambda_n|)}{\ln n}<2.
\end{equation}
The next corollary shows that one can recover the polynomial potential  from the asymptotic formula for the
eigenvalues.
\begin{corollary}\label{cor6}
Suppose that $\gcd(m,\ell)=1$ and that $\{\lambda_n\}_{n\geq N_0}$ are the eigenvalues of $H_{\ell}$ with the
potential $V$. If
\begin{equation}\label{bohr_eq}
\sum_{j=0}^{m+1}c_j^*
\lambda_n^{\frac{1}{2}+\frac{1-j}{m}}+O\left(\lambda_n^{-\frac{1}{2}-\frac{1}{m}}\right){=}\left(n+\frac{1}{2}\right)
\end{equation}
as $n\to\infty$ for some $c_j^*\in\C$, or if
\begin{equation}\label{bohr_eq1}
\lambda_n=\sum_{j=0}^{m+1}d_j^*\left(n+\frac{1}{2}\right)^{\frac{2m}{m+2}\left(1-\frac{j}{m}\right)}+O\left(n^{-\frac{4}{m+2}}\right),
\end{equation}
for some $d_j^*\in\C$, then  there exists $V_0$ such that for each polynomial $V$ with which $H_{\ell}$ generates
$\{\lambda_n\}_{n\geq N_0}$, $V(z)=V_0(z-z_0)$ for some $z_0\in\C$.
\end{corollary}

\begin{proof}
The coefficients $c_{\ell,j}( {a})$ and  $d_{\ell,j}( {a})$  have the following properties.
\begin{itemize}
\item[]{(i)} The $c_{\ell,j}( {a})$ and  $d_{\ell,j}( {a})$ are all {\it real} polynomials in terms of the coefficients $ {a}$ of $P(x)$.
\item[]{(ii)} The coefficients $c_{\ell,0}( {a})$ and  $d_{\ell,0}( {a})$ do not depend on $ {a}$ (they are  constants).
\item[]{(iii)} For $1\leq j\leq m$, the polynomials $c_{\ell,j}( {a})$ and  $d_{\ell,j}( {a})$ depend only on   $a_1,\,a_2,\dots,a_{j}$.
Furthermore, if $\gcd(m,\ell)=1$ and if $j\not=1,\,\frac{m}{2}+1$, then $c_{\ell,j}( {a})$ and  $d_{\ell,j}( {a})$ are
non-constant linear functions of $a_j$.
\end{itemize}
Here we will sketch the proof when \eqref{bohr_eq} is assumed. First, notice that $c_1^*=0$. Otherwise,  $\lambda_n$
are not the eigenvalues of $H_{\ell}$, according to  \eqref{dmlj_def}. Choose $a_1=0$. Then since $c_{\ell,2}(
{a})$ is a non-constant linear function of $a_2$ and depends on $a_1$ and $a_2$, $c_{\ell, 2}( {a})=c_2^*$
determines $a_2$. If $m$ is even, $c_{\ell,\frac{m}{2}+1}(a)=0$ and we use $\eta_{\ell}(a)$ instead of
$c_{\ell,\frac{m}{2}+1}(a)$. Suppose that for $2\leq j\leq m-1$, the $c_2^*,\dots,c_j^*$ uniquely determine
$a_2,\, a_3,\dots, a_j$. Then by (iii), $c_{\ell,j+1}( {a})=c_{j+1}^*$ determines $a_{j+1}$ uniquely. So by
induction, one can determine all $a_j$ for $2\leq j\leq m$ and $V_0$ is the potential with these $a_j$'s. Then
$V(z)=V_0(z-z_0)$ for some $z_0\in\C$. Otherwise, the eigenvalues $\{\lambda_n\}$ do not satisfy \eqref{bohr_eq}.

The case when \eqref{bohr_eq1} holds can be handled similarly.
\end{proof}

Next, we will prove Theorems \ref{thm_111} and \ref{thm_112}.
\begin{proof}[Proof of Theorems ~\ref{thm_111}]
If $V(z-z_0)$ for some $z_0\in\C$ is $\mathcal{PT}$-symmetric, then by Theorem \ref{main_thm2}, all but finitely many eigenvalues are real and hence, there are infinitely many real eigenvalues.

Suppose that $H_{\ell}$ with the potential $V(z)$ has infinitely many real eigenvalues. Then we can always find $z_0\in\C$ so that $V(z-z_0)$ has no $z^{m-1}$-term, that is, $a_1=0$. We will show that $V(z-z_0)$ is $\mathcal{PT}$-symmetric. Since there are infinitely many real eigenvalues, $c_{\ell,j}(a)$, $0\leq j\leq m+1$, and $\eta_{\ell}$ in \eqref{main_result} are real. Also, the sine term in \eqref{dmlj_def} does not vanish except for $j=1$. And the cosine term does not vanish except for $j=\frac{m}{2}+1$ when $m$ is even in which $\eta_{\ell}$ replaces the role of $c_{\ell,\frac{m}{2}}$.

Like we did for proof of Corollary \ref{cor6}, since $a_1=0$, from the above properties of $c_{\ell,j}$ and $\eta_{\ell}$,  by the induction, we can show that $a_j$ for $2\leq j\leq m$ are all real and hence, $V(z-z_0)$ is $\mathcal{PT}$-symmetric.
\end{proof}
\begin{proof}[Proof of Theorems ~\ref{thm_112}]
Suppose that
$\lambda_n-\widetilde{\lambda}_n=o\left(1\right)$ as $n\to\infty$. Then $d_{\ell,j}(a)=d_{\ell,j}(\widetilde{a})$ for all $1\leq j\leq m+1$. Then Corollary \ref{cor6} completes  proof.
\end{proof}
\section{Properties of the solutions}
\label{prop_sect}
In this section, we introduce work of Hille \cite{Hille} and Sibuya \cite{Sibuya} about properties of the solutions of \eqref{ptsym}.

First, we scale equation \eqref{ptsym} because many facts that we
need later are stated for the scaled equation. Let $u$ be a
solution of (\ref{ptsym}) and let $v(z)=u(-iz)$.
 Then $v$ solves
\begin{equation}\label{rotated1}
-v^\dd(z)+[(-1)^{\ell+1}z^m+P(z)+\lambda]v(z)=0.
\end{equation}
When $\ell$ is odd, \eqref{rotated1} becomes
\begin{equation}\label{rotated}
-v^\dd(z)+[z^m+P(z)+\lambda]v(z)=0.
\end{equation}
Later we will handle the case when $\ell$ is even.

Since we scaled the argument of $u$, we must rotate the boundary conditions. We state them in a more general
context by using the following definition.
\begin{definition}
{\rm {\it The Stokes sectors} $S_k$ of the equation (\ref{rotated})
are
$$ S_k=\left\{z\in \C:\left|\arg (z)-\frac{2k\pi}{m+2}\right|<\frac{\pi}{m+2}\right\}\quad\text{for}\quad k\in \Z.$$ }
\end{definition}
See Figure \ref{f:graph1}.
\begin{figure}[t]
    \begin{center}
    \includegraphics[width=.4\textwidth]{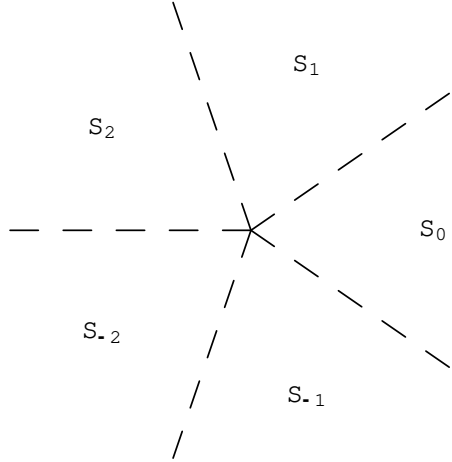}
    \end{center}
 \vspace{-.5cm}
\caption{The Stokes sectors for $m=3$. The dashed rays represent $\arg z=\pm\frac{\pi}{5},\,\pm\frac{3\pi}{5},\, \pi.$}\label{f:graph1}
\end{figure}

It is known from Hille \cite[\S 7.4]{Hille} that every nonconstant solution of (\ref{rotated}) either decays to
zero or blows up exponentially, in each Stokes sector $S_k$.
\begin{lemma}[\protect{\cite[\S 7.4]{Hille}}]\label{gen_pro}
${} $

\begin{itemize}
\item [(i)] For each $k\in\Z$, every solution $v$ of (\ref{rotated})  is asymptotic to
\begin{equation} \label{asymp-formula}
(const.)z^{-\frac{m}{4}}\exp\left[\pm \int^z \left[\xi^m+P(\xi)+\lambda\right]^{\frac{1}{2}}\,d\xi\right]
\end{equation}
as $z \rightarrow \infty$ in every closed subsector of $S_k$.

\item [(ii)] If a nonconstant solution $v$ of \eqref{rotated} decays in $S_k$, it must blow up
in $S_{k-1}\cup S_{k+1}$. However, when $v$ blows up in $S_k$, $v$ need not be decaying in $S_{k-1}$ or in $S_{k+1}$.
\end{itemize}
\end{lemma}
Lemma \ref{gen_pro} (i) implies that if $v$ decays along one ray in $S_k$, then it decays along all rays in $S_k$.
Also, if $v$ blows up along one ray in $S_k$, then it blows up along all rays in $S_k$. Thus, the boundary
conditions \eqref{bdcond} with $1\leq\ell\leq m-1$ represent all decaying boundary conditions.

Still with $\ell$ odd, the two rays in \eqref{bdcond} map, by $z\mapsto -iz$,  to the rays
$\arg(z)=\pm\frac{(\ell+1)\pi}{m+2}$ which are the center rays of the Stokes sectors $S_{\frac{\ell+1}{2}}$ and
$S_{-\frac{\ell+1}{2}}$ and the boundary conditions \eqref{bdcond} on $u$ become
\begin{equation}\nonumber
\text{$v$ decays to zero in the Stokes sector $S_{\frac{\ell+1}{2}}$ and $S_{-\frac{\ell+1}{2}}$.}
\end{equation}

When $\ell$ is even, we let
 $y(z)=v(\omega^{-\frac{1}{2}}z)$ so that \eqref{rotated1} becomes
\begin{equation}\label{rotated2}
-y^\dd(z)+[z^m+\omega^{-1} P(\omega^{-\frac{1}{2}}z)+\omega^{-1}\lambda]y(z)=0,
\end{equation}
 where
$$\omega=\exp\left[\frac{2\pi i}{m+2}\right]$$
and hence, $\omega^{-{\frac{m}{2}+1}}=-1$. For these cases, the boundary conditions \eqref{bdcond} become
\begin{equation}\nonumber
\text{$y$ decays to zero in the Stokes sector $S_{\frac{\ell+2}{2}}$ and $S_{-\frac{\ell}{2}}$.}
\end{equation}

The following theorem is a special case of Theorems 6.1, 7.2, 19.1 and 20.1 of Sibuya \cite{Sibuya} that is the
main ingredient of the proofs of the main results in this paper. For this we will use $r_m=-\frac{m}{4}$ if $m$ is
odd, and $r_m=-\frac{m}{4}-b_{\frac{m}{2}+1}(a)$ if $m$ is even.
\begin{theorem}\label{prop}
Equation (\ref{rotated}), with $a\in \C^{m}$, admits a solution  $f(z,a,\lambda)$ with the following properties.
\begin{enumerate}
\item[(i)] $f(z,a,\lambda)$ is an entire function of $z,a $ and $\lambda$.
\item[(ii)] $f(z,a,\lambda)$ and $f^\d(z,a,\lambda)=\frac{\partial}{\partial z}f(z,a,\lambda)$ admit
the following asymptotic expansions. Let $\varepsilon>0$. Then
\begin{align}
f(z,a,\lambda)=&\qquad z^{r_m}(1+O(z^{-1/2}))\exp\left[-F(z,a,\lambda)\right],\nonumber\\
f^\d(z,a,\lambda)=&-z^{r_m+\frac{m}{2}}(1+O(z^{-1/2}))\exp\left[-F(z,a,\lambda) \right],\nonumber
\end{align}
as $z$ tends to infinity in  the sector $|\arg z|\leq \frac{3\pi}{m+2}-\varepsilon$, uniformly on each compact set of $(a,\lambda)$-values .
Here
\begin{equation}\nonumber
F(z,a,\lambda)=\frac{2}{m+2}z^{\frac{m}{2}+1}+\sum_{1\leq j<\frac{m}{2}+1}\frac{2}{m+2-2j}b_j(a) z^{\frac{1}{2}(m+2-2j)}.
\end{equation}
\item[(iii)] Properties \textup{(i)} and \textup{(ii)} uniquely determine the solution $f(z,a,\lambda)$ of (\ref{rotated}).
\item[(iv)] For each fixed $a\in\C^{m}$ and $\delta>0$, $f$ and $f^\d$ also admit the asymptotic expansions,
\begin{align}
f(0,a,\lambda)=&[1+O\left(\lambda^{-\rho}\right)]\lambda^{-1/4}\exp\left[L(a,\lambda)\right],\label{eq1}\\
f^\d(0,a,\lambda)=&-[1+O\left(\lambda^{-\rho}\right)]\lambda^{1/4}\exp\left[L(a,\lambda)\right],\label{eq2}
\end{align}
as $\lambda\to\infty$ in the sector $|\arg(\lambda)|\leq\pi-\delta$, uniformly on each compact set of
$a\in\C^{m}$, where
\begin{align}
L(a,\lambda)=\left\{
                    \begin{array}{rl}
                    &\int_0^{+\infty}\left(\sqrt{t^m+P(t)+\lambda}- t^{\frac{m}{2}}-\sum_{j=1}^{\frac{m+1}{2}}b_j(a)t^{\frac{m}{2}-j}\right)\,dt \quad \text{if $m$ is odd,}\\
                   &\int_0^{+\infty}\left(\sqrt{t^m+P(t)+\lambda}- t^{\frac{m}{2}}-\sum_{j=1}^{\frac{m}{2}}b_j(a)t^{\frac{m}{2}-j}-\frac{b_{\frac{m}{2}+1}}{t+1}\right)\,dt  \quad \text{if $m$ is even.}
                    \end{array}
                         \right. \nonumber
\end{align}
\item[(v)] The entire functions  $\lambda\mapsto f(0,a,\lambda)$ and $\lambda\mapsto f^\d(0,a,\lambda)$ have orders $\frac{1}{2}+\frac{1}{m}$.
\end{enumerate}
\end{theorem}
\begin{proof}
In  Sibuya's book \cite{Sibuya}, see Theorem 6.1 for a proof of (i) and (ii); Theorem 7.2 for a proof of (iii) with the error terms $o(1)$;
and Theorem 19.1 for a proof of (iv).  Moreover, (v) is a consequence of (iv) along with Theorem 20.1 in
\cite{Sibuya}. The error terms in (iii) are improved from $o(1)$ to $O\left(\lambda^{-\rho}\right)$ in \cite{Shin3}. Note that properties (i), (ii) and (iv) are summarized on pages 112--113 of Sibuya \cite{Sibuya}.
\end{proof}

\begin{remarks}
{\rm Throughout this paper, we will deal with numbers like
$\left(\omega^{\alpha}\lambda\right)^{s}$ for some $s\in\R$ and $\alpha\in\C$. As usual, we will use
$$\omega^{\alpha}=\exp\left[\alpha \frac{2\pi i}{m+2}\right]$$
and if $\arg(\lambda)$ is specified, then
$$\arg\left(\left(\omega^{\alpha}\lambda\right)^{s}\right)=s\left[\arg(\omega^{\alpha})+\arg(\lambda)\right]=s\left[\Re(\alpha)\frac{2\pi}{m+2}+\arg(\lambda)\right],\quad s\in\R.$$
}
\end{remarks}

\begin{lemma}\label{asy_lemma}
Let $m\geq 3$ and $a\in\C^{m}$ be fixed. Then
\begin{equation}\label{L_def}
L(a,\lambda)=\sum_{j=0}^{\infty}K_{m,j}(a)\lambda^{\frac{1}{2}+\frac{1-j}{m}}-\frac{\nu(a)}{m}\ln(\lambda)
\end{equation}
as $\lambda\to\infty$ in the sector $|\arg(\lambda)|\leq \pi-\delta$, uniformly on each compact set of
$a\in\C^{m}$.
\end{lemma}
\begin{proof}
See \cite{Shin2} for a proof.
\end{proof}
Sibuya \cite{Sibuya} proved the following corollary, directly from Theorem \ref{prop}, that will be used later in
Sections \ref{sec_4} and \ref{sec_5}.
\begin{corollary}\label{lemma_decay}
Let $a\in\C^{m}$ be fixed. Then $L(a,\lambda)=K_{m}\lambda^{\frac{1}{2}+\frac{1}{m}}(1+o(1))$ as $\lambda$ tends
to infinity in the sector $|\arg \lambda|\leq \pi-\delta$, and hence
\begin{equation}\label{re_part}
\Re
\left(L(a,\lambda)\right)=K_{m}\cos\left(\frac{m+2}{2m}\arg(\lambda)\right)|\lambda|^{\frac{1}{2}+\frac{1}{m}}(1+o(1))
\end{equation}
as $\lambda\to\infty$ in the sector $|\arg (\lambda)|\leq \pi-\delta$.

In particular, $\Re \left(L(a,\lambda)\right)\to+\infty$ as $\lambda\to\infty$ in any closed subsector of the
sector $|\arg(\lambda)|<\frac{m\pi}{m+2}$. In addition, $\Re \left(L(a,\lambda)\right)\to-\infty$ as
$\lambda\to\infty$ in any closed subsector of the sectors $\frac{m\pi}{m+2}<|\arg(\lambda)|<\pi-\delta$.
\end{corollary}

Based on Corollary \ref{lemma_decay}, Sibuya \cite[Thm.\ 29.1]{Sibuya} also computed the leading term  in
\eqref{main_result} for $\ell=1$. Also, Sibuya \cite{Sibuya} constructed solutions of \eqref{rotated} that decays
in $S_k$, $k\in\Z$. Before we introduce this, we let
\begin{equation}\label{G_def}
G^{\ell}(a):=(\omega^{(m+1)\ell}a_1, \omega^{m\ell}a_2,\ldots,\omega^{2\ell}a_{m})\quad \text{for}\quad \ell\in
\frac{1}{2}\Z.
\end{equation}
Then we have the following lemma, regarding some properties of $G^{\ell}(\cdot)$.
\begin{lemma}\label{lemma_25}
For $a\in\C^{m}$ fixed, and $\ell_1,\ell_2,\ell\in\frac{1}{2}\Z$,
$G^{\ell_1}(G^{\ell_2}(a))=G^{\ell_1+\ell_2}(a)$, and
\begin{equation}\nonumber
b_{j,k}(G^{\ell}(a))=\omega^{((m+2)k-j)\ell}b_{j,k}(a),\quad \ell\in\frac{1}{2}\Z.
\end{equation}
In particular,
\begin{equation}\nonumber
b_j(G^{\ell}(a))=\omega^{-j\ell}b_j(a),\,\, \ell\in\Z.
\end{equation}
\end{lemma}



Next, recall that the function
 $f(z,a,\lambda)$  in Theorem \ref{prop} solves (\ref{rotated}) and decays to zero exponentially as $z\rightarrow \infty$ in
  $S_0$, and  blows up in $S_{-1}\cup S_1$. One can check that the function
$$f_k(z,a,\lambda):=f(\omega^{-k}z,G^k(a),\omega^{2k}\lambda),\quad k\in\Z,$$
 which is obtained by scaling $f(z,G^k(a),\omega^{2k}\lambda)$ in the $z$-variable, also solves (\ref{rotated}). It is clear
 that $f_0(z,a,\lambda)=f(z,a,\lambda)$, and that $f_k(z,a,\lambda)$ decays in $S_k$ and blows up in $S_{k-1}\cup S_{k+1}$ since
 $f(z,G^k(a),\omega^{2k}\lambda)$ decays in $S_0$. Since no nonconstant solution decays in two consecutive Stokes sectors (see Lemma \ref{gen_pro}
 (ii)), $f_{k}$ and $f_{k+1}$ are linearly independent and hence any solution of (\ref{rotated}) can be expressed as a linear combination of these two.
 Especially,  for each $k\in\Z$ there exist some coefficients $C_k(a,\lambda)$ and $\widetilde{C}_k(a,\lambda)$ such that
\begin{equation}\label{stokes}
f_{k}(z,a,\lambda)=C_k(a,\lambda)f_{0}(z,a,\lambda)+\widetilde{C}_k(a,\lambda)f_{-1}(z,a,\lambda).
\end{equation}
We then see that
\begin{equation}\label{C_def}
C_k(a,\lambda)=-\frac{\mathcal{W}_{k,-1}(a,\lambda)}{\mathcal{W}_{-1,0}(a,\lambda)}\quad\text{and}\quad
\widetilde{C}_k(a,\lambda)=\frac{\mathcal{W}_{k,0}(a,\lambda)}{\mathcal{W}_{-1,0}(a,\lambda)},
\end{equation}
where $\mathcal{W}_{j,\ell}=f_jf_{\ell}^\d -f_j^\d f_{\ell}$ is the Wronskian of $f_j$ and $f_{\ell}$. Since both
$f_j,\,f_{\ell}$ are solutions of the same linear equation (\ref{rotated}), we know that the Wronskians are
constant functions of $z$. Also, $f_k$ and $f_{k+1}$ are linearly independent, and hence
$\mathcal{W}_{k,k+1}\not=0$ for all $k\in \Z$.

Also, the following is an easy consequence of \eqref{stokes} and \eqref{C_def}. For each $k,\ell\in\Z$ we have
\begin{align}
\mathcal{W}_{\ell,k}(a,\lambda)&=C_k(a,\lambda)\mathcal{W}_{\ell,0}(a,\lambda)+\widetilde{C}_k(a,\lambda)\mathcal{W}_{\ell,-1}(a,\lambda)\nonumber\\
&=-\frac{\mathcal{W}_{k,-1}(a,\lambda)\mathcal{W}_{\ell,0}(a,\lambda)}{\mathcal{W}_{-1,0}(a,\lambda)}+\frac{\mathcal{W}_{k,0}(a,\lambda)\mathcal{W}_{\ell,-1}(a,\lambda)}{\mathcal{W}_{-1,0}(a,\lambda)}.\label{kplus}
\end{align}

Moreover, we have the following lemma that is useful  later on.
\begin{lemma}\label{shift_lemma}
Suppose $k,\,j\in\Z$. Then
\begin{equation}\label{kplus1}
\mathcal{W}_{k+1,j+1}(a,\lambda)=\omega^{-1}\mathcal{W}_{k,j}(G(a),\omega^2\lambda),
\end{equation}
and $\mathcal{W}_{0,1}(a,\lambda)=2\omega^{\mu(a)}$, where $\mu(a)=\frac{m}{4}-\nu(a)$.
\end{lemma}
\begin{proof}
See Sibuya \cite[pp.\ 116-118]{Sibuya}.
\end{proof}


\section{Asymptotics of $\mathcal{W}_{-1,1}(a,\lambda)$}\label{sec_4}
In this section, we introduce asymptotic expansions of $\mathcal{W}_{-1,1}(a,\lambda)$ as $\lambda\to\infty$ along
the rays in the complex plane \cite{Shin2}.

First, we provide an asymptotic expansion of  the Wronskian  $\mathcal{W}_{0,j}(a,\lambda)$ of $f_0$ and $f_j$
that will be frequently used later.
\begin{lemma}\label{lemma7}
Suppose that $1\leq j\leq \frac{m}{2}+1$. Then for each $a\in\C^{m}$,
\begin{equation}\label{sec_eq1}
\mathcal{W}_{0,j}(a,\lambda)=[2i\omega^{-\frac{j}{2}}+O\left(\lambda^{-\rho}\right)]\exp\left[L(G^{j}(a),\omega^{2j-m-2}\lambda)+L(a,\lambda)\right],
\end{equation}
as $\lambda\to\infty$ in the sector
\begin{equation}\label{sector0}
-\pi+\delta\leq \pi-\frac{4j\pi}{m+2}+\delta \leq \arg(\lambda)\leq \pi-\delta.
\end{equation}
\end{lemma}

Next, we provide an asymptotic expansion of $\mathcal{W}_{-1,1}(a,\lambda)$ as $\lambda\to\infty$ in the sector
near the negative real axis.
\begin{theorem}\label{thm_neg}
Let $m\geq 3$, $a\in\C^{m}$ and $0<\delta<\frac{\pi}{m+2}$ be fixed. Then
\begin{equation}\label{asy_1}
\mathcal{W}_{-1,1}(a,\lambda)=[2i+O\left(\lambda^{-\rho}\right)]\exp\left[L(G^{-1}(a),\omega^{-2}\lambda)+L(G(a),\omega^{-m}\lambda)\right],
\end{equation}
as $\lambda\to \infty$ along the rays in the sector
\begin{equation}\label{sector1}
\pi-\frac{4\pi}{m+2}+\delta\leq \arg(\lambda)\leq \pi+\frac{4\pi}{m+2}-\delta.
\end{equation}
\end{theorem}
\begin{proof}
This is an easy consequence of Lemma \ref{lemma7} with $j=2$ and \eqref{kplus1}.
\end{proof}

Also,  for integers $m\geq 4$ we provide an asymptotic expansion of $\mathcal{W}_{-1,1}(a,\lambda)$ as
$\lambda\to\infty$ in the sector $|\arg(\lambda)|\leq \pi-\delta$.
\begin{theorem}\label{zero_thm}
Let $a\in\C^{m}$ and $0<\delta<\frac{\pi}{2(m+2)}$ be fixed. If $m\geq 4$ then
\begin{align}
\mathcal{W}_{-1,1}(a,\lambda)=&[2\omega^{\frac{1}{2}+\mu(a)}+O\left(\lambda^{-\rho}\right)]\exp\left[L(G^{-1}(a),\omega^{-2}\lambda)-L(a,\lambda)\right]\nonumber\\
&+[2\omega^{\frac{1}{2}+\mu(a)+2\nu(a)}+O\left(\lambda^{-\rho}\right)]\exp\left[L(G(a),\omega^{2}\lambda)-L(a,\lambda)\right],\label{tot_asy}
\end{align}
as $\lambda\to\infty$ in the sector
\begin{equation}\label{sector41}
-\pi+\delta \leq \arg(\lambda)\leq \pi-\delta.
\end{equation}
\end{theorem}

Next, we provide an asymptotic expansion of $\mathcal{W}_{-1,1}(a,\lambda)$ as $\lambda\to\infty$ along the rays
in the upper- and lower- half planes.
\begin{corollary}\label{lemma_up}
Let $m\geq 4$, $a\in\C^{m}$ and $0<\delta<\frac{\pi}{m+2}$ be fixed. Then
\begin{equation}\nonumber
\mathcal{W}_{-1,1}(a,\lambda)=[2\omega^{\frac{1}{2}+\mu(a)}+O\left(\lambda^{-\rho}\right)]\exp\left[L(G^{-1}(a),\omega^{-2}\lambda)-L(a,\lambda)\right],
\end{equation}
as $\lambda\to\infty$ in the sector $\delta \leq \arg(\lambda)\leq \pi-\delta$.
Also,
\begin{equation}\nonumber
\mathcal{W}_{-1,1}(a,\lambda)=[2\omega^{\frac{1}{2}+\mu(a)+2\nu(a)}+O\left(\lambda^{-\rho}\right)]\exp\left[L(G(a),\omega^{2}\lambda)-L(a,\lambda)\right],
\end{equation}
as $\lambda\to\infty$ in the sector $-\pi+\delta \leq \arg(\lambda)\leq -\delta$.
\end{corollary}
\begin{proof}
We will determine which term in \eqref{tot_asy} dominates in the upper and lower half planes.

Since, by \eqref{re_part},
\begin{equation}\nonumber
\Re(L(a,\lambda))=K_m\cos\left(\frac{m+2}{2m}\arg(\lambda)\right)|\lambda|^{\frac{1}{2}+\frac{1}{m}}(1+o(1)),
\end{equation}
we have
\begin{align}
&\left[\Re(L(G^{-1}(a),\omega^{-2}\lambda))-\Re(L(a,\lambda))\right]-\left[\Re(L(G(a),\omega^{2}\lambda))-\Re(L(a,\lambda))\right]\nonumber\\
&=K_m\left[\cos\left(-\frac{2\pi}{m}+\frac{m+2}{2m}\arg(\lambda)\right)-\cos\left(\frac{2\pi}{m}+\frac{m+2}{2m}\arg(\lambda)\right)\right]|\lambda|^{\frac{1}{2}+\frac{1}{m}}(1+o(1))\nonumber\\
&=2K_m\sin\left(\frac{2\pi}{m}\right)\sin\left(\frac{m+2}{2m}\arg(\lambda)\right)|\lambda|^{\frac{1}{2}+\frac{1}{m}}(1+o(1)).\nonumber
\end{align}
Thus, the first term in \eqref{tot_asy} dominates as $\lambda\to\infty$ along the rays in the upper half plane, and the second term dominates in the lower half plane. This completes the proof.
\end{proof}
\begin{proof}[Proof of Theorem ~\ref{zero_thm}]
In \cite{Shin2}, $C(a,\lambda)$ is used for $\frac{\mathcal{W}_{-1,1}(a,\lambda)}{\mathcal{W}_{0,1}(a,\lambda)}$
and asymptotics of $C(a,\lambda)$ are provided.  Notice that
$\mathcal{W}_{-1,1}(a,\lambda)=2\omega^{\mu(a)}C(a,\lambda)$.

Theorem 13 in \cite{Shin2} implies
 \eqref{tot_asy}  for the sector
\begin{equation}\label{sector4}
\pi-\frac{4\lfloor\frac{m}{2}\rfloor\pi}{m+2}+\delta \leq \arg(\lambda)\leq \pi-\frac{4\pi}{m+2}-\delta.
\end{equation}
Theorem 14 in \cite{Shin2} implies that
\begin{align}
\mathcal{W}_{-1,1}(a,\lambda)=&[2\omega^{\frac{1}{2}+\mu(a)}+O\left(\lambda^{-\rho}\right)]\exp\left[L(G^{-1}(a),\omega^{-2}\lambda)-L(a,\lambda)\right]\nonumber\\
&+[2\omega^{1+2\mu(a)+4\nu(a)}+O\left(\lambda^{-\rho}\right)]\exp\left[-L(G^2(a),\omega^{2-m}\lambda)-L(a,\lambda)\right],\nonumber
\end{align}
as $\lambda\to\infty$ in the sector $\pi-\frac{8\pi}{m+2}+\delta\leq\arg(\lambda)\leq\pi-\delta$. One can check that the first term dominates in this sector, by using an argument  similar to that in the proof of Corollary \ref{lemma_up}.

Also, Theorem 15 in \cite{Shin2} implies that
\begin{align}
\mathcal{W}_{-1,1}(a,\lambda)=&[2\omega^{1+2\mu(a)}+O\left(\lambda^{-\rho}\right)]\exp\left[-L(a,\omega^{-m-2}\lambda)-L(^{-2}(a),\omega^{-4}\lambda)\right]\nonumber\\
&+[2\omega^{\frac{1}{2}+\mu(a)+2\nu(a)}+O\left(\lambda^{-\rho}\right)]\exp\left[L(G(a),\omega^{-m}\lambda)-L(a,\omega^{-m-2}\lambda)\right],\nonumber
\end{align}
as $\lambda\to\infty$ in the sector $\pi+\delta\leq\arg(\lambda)\leq\pi+\frac{8\pi}{m+2}-\delta$. One can check that the second term dominates in this sector. Then we replace $\lambda$ by $\omega^{m+2}\lambda$ to convert the sector here to $-\pi+\delta\leq\arg(\lambda)\leq-\pi+\frac{8\pi}{m+2}-\delta$. This completes the proof.
\end{proof}

\begin{theorem}\label{thm_sector2}
Let $m=3$ and let $a\in\C^{m}$ and $0<\delta<\frac{\pi}{m+2}$ be fixed. Then
\begin{align}
\mathcal{W}_{-1,1}(a,\lambda)=&[-2\omega^{-\frac{5}{4}}+O\left(\lambda^{-\rho}\right)]\exp\left[L(G^{4}(a),\omega^{-2}\lambda)-L(a,\lambda)\right]\nonumber\\
&-[2i\omega^{\frac{5}{2}}+O\left(\lambda^{-\rho}\right)]\exp\left[-L(G^2(a),\omega^{-1}\lambda)-L(a,\lambda)\right],\nonumber
\end{align}
as $\lambda\to\infty$ in the sector
$-\delta \leq \arg(\lambda)\leq \pi-\delta$.
Also,
\begin{align}
\mathcal{W}_{-1,1}(a,\lambda)=&[-2i\omega^{\frac{5}{2}}+O\left(\lambda^{-\rho}\right)]\exp\left[-L(a,\omega^{-5}\lambda)-L(G^{-2}(a),\omega^{-4}\lambda)\right]\nonumber\\
&+[2\omega^{\frac{15}{4}}+O\left(\lambda^{-\rho}\right)]\exp\left[L(G(a),\omega^{-3}\lambda)-L(a,\omega^{-5}\lambda)\right],\nonumber
\end{align}
as $\lambda\to\infty$ in the sector $-\pi+\delta \leq \arg(\lambda)\leq \delta$.
\end{theorem}
\begin{proof}
See Theorems 14 and 15 in \cite{Shin2} for a proof.
\end{proof}

\section{Asymptotics of $\mathcal{W}_{-1,n}(a,\lambda)$}\label{sec_5}
In this section, we will provide asymptotic expansions of $\mathcal{W}_{-1,n}(a,\cdot)$, zeros of which will be
closely related with the eigenvalues of $H_{n}$.

First, we treat the cases when  $1\leq n <\lfloor\frac{m}{2}\rfloor$.
\begin{theorem}
Let $1\leq n <\lfloor\frac{m}{2}\rfloor$ be an integer. Then $\mathcal{W}_{-1,n}(a,\cdot)$  admits the following
asymptotic expansion
\begin{equation}\label{up_asy}
\mathcal{W}_{-1,n}(a,\lambda)=[2\omega^{\frac{2-n}{2}+\mu(G^{n-1}(a))}+O\left(\lambda^{-\rho}\right)]\exp\left[L(G^{-1}(a),\omega^{-2}\lambda)-L(G^{n-1}(a),\omega^{2(n-1)}\lambda)\right],
\end{equation}
as $\lambda\to\infty$ in the sector
\begin{equation}\label{zero3_sector}
-\frac{2(n-1)\pi}{m+2}+\delta\leq\arg(\lambda)\leq\pi -\frac{4n\pi}{m+2}+\delta.
\end{equation}

Also,
\begin{align}
\mathcal{W}_{-1,n}(a,\lambda)&=[2\omega^{\frac{2-n}{2}+\mu(G^{n-1}(a))}+O\left(\lambda^{-\rho}\right)]\exp\left[L(G^{-1}(a),\omega^{-2}\lambda)-L(G^{n-1}(a),\omega^{2(n-1)}\lambda)\right]\nonumber\\
&+[2\omega^{\frac{2-n}{2}+\mu(G^{-1}(a))}+O\left(\lambda^{-\rho}\right)]\exp\left[L(G^n(a),\omega^{2n}\lambda)-L(a,\lambda)\right],\label{zero_asy}
\end{align}
as $\lambda\to\infty$ in the sector
\begin{equation}\label{zero2_sector}
-\frac{2(n-1)\pi}{m+2}-\delta\leq\arg(\lambda)\leq-\frac{2(n-1)\pi}{m+2}+\delta.
\end{equation}
\end{theorem}
\begin{proof}
First we will prove \eqref{up_asy} for the sector
 \begin{equation}\label{zero1_sector}
-\frac{2(n-1)\pi}{m+2}+\delta\leq\arg(\lambda)\leq\pi -\frac{4n\pi}{m+2}-\delta
\end{equation}
 and the second part of the theorem by induction on $n$.

The case when $n=1$ is trivially satisfied by  Theorem \ref{zero_thm} and Corollary \ref{lemma_up} since $\mu(a)+2\nu(a)=\mu(G^{-1}(a))$.

Suppose that \eqref{up_asy} holds in the sector \eqref{zero1_sector} for $n-1$. From this induction hypothesis we have
\begin{align}
\mathcal{W}_{0,n}(a,\lambda)&=\omega^{-1}\mathcal{W}_{-1,n-1}(G(a),\omega^2\lambda)\nonumber\\
&=[2\omega^{-\frac{n-1}{2}+\mu(G^{n-1}(a))}+O\left(\lambda^{-\rho}\right)]\exp\left[L(a,\lambda)-L(G^{n-1}(a),\omega^{2(n-1)}\lambda)\right],\label{1steq}
\end{align}
as $\lambda\to\infty$ in the sector
\begin{equation}\nonumber
-\frac{2(n-2)\pi}{m+2}+\delta\leq\arg(\omega^2\lambda)\leq\pi -\frac{4(n-1)\pi}{m+2}-\delta,
\end{equation}
that is,
\begin{equation}\label{dom_sector}
-\frac{2n\pi}{m+2}+\delta\leq\arg(\lambda)\leq\pi -\frac{4n\pi}{m+2}-\delta.
\end{equation}

Also, from Lemma \ref{lemma7} if $1\leq j\leq\frac{m}{2}+1$, then we have
\begin{equation}\label{2ndeq}
\mathcal{W}_{0,j}(a,\lambda)=[2i\omega^{-\frac{j}{2}}+O\left(\lambda^{-\rho}\right)]\exp\left[L(G^{j}(a),\omega^{2j-m-2}\lambda)+L(a,\lambda)\right],
\end{equation}
as $\lambda\to\infty$ in the sector
\begin{equation}\nonumber
 \pi-\frac{4j\pi}{m+2}+\delta \leq \arg(\lambda)\leq \pi-\delta.
\end{equation}
We solve \eqref{kplus} for $\mathcal{W}_{\ell,-1}(a,\lambda)$ and set $\ell=n$ to get
\begin{equation}\label{kn_asy}
\mathcal{W}_{-1,n}(a,\lambda)=\frac{\mathcal{W}_{-1,0}(a,\lambda)\mathcal{W}_{n,k}(a,\lambda)}{\mathcal{W}_{0,k}(a,\lambda)}+\frac{\mathcal{W}_{0,n}(a,\lambda)\mathcal{W}_{-1,k}(a,\lambda)}{\mathcal{W}_{0,k}(a,\lambda)}
\end{equation}
Set $k=\lfloor\frac{m}{2}\rfloor$. Then since $1\leq k-n<k=\lfloor\frac{m}{2}\rfloor$, using \eqref{kplus1},
\begin{align}
\mathcal{W}_{-1,n}(a,\lambda)&=\frac{2\omega^{\mu(G^{-1}(a))}\mathcal{W}_{0,k-n}(G^n(a),\omega^{2n}\lambda)}{\omega^{n-1}\mathcal{W}_{0,k}(a,\lambda)}+\frac{\mathcal{W}_{0,n}(a,\lambda)\mathcal{W}_{0,k+1}(G^{-1}(a),\omega^{-2}\lambda)}{\omega^{-1}\mathcal{W}_{0,k}(a,\lambda)}\nonumber\\
&=\frac{2\omega^{\mu(G^{-1}(a))}[2i\omega^{-\frac{k-n}{2}}+O\left(\lambda^{-\rho}\right)]\exp\left[L(G^{n}(a),\omega^{2k-m-2}\lambda)+L(G^n(a),\omega^{2n}\lambda)\right]}{\omega^{n-1}[2i\omega^{-\frac{k}{2}}+O\left(\lambda^{-\rho}\right)]\exp\left[L(G^{k}(a),\omega^{2k-m-2}\lambda)+L(a,\lambda)\right]}\nonumber\\
&+\frac{[2\omega^{-\frac{n-1}{2}+\mu(G^{n-1}(a))}+O\left(\lambda^{-\rho}\right)]\exp\left[L(a,\lambda)-L(G^{n-1}(a),\omega^{2(n-1)}\lambda)\right]}{\omega^{-1}[2i\omega^{-\frac{k}{2}}+O\left(\lambda^{-\rho}\right)]\exp\left[L(G^{k}(a),\omega^{2k-m-2}\lambda)+L(a,\lambda)\right]}\nonumber\\
&\times [2i\omega^{-\frac{k+1}{2}}+O\left(\lambda^{-\rho}\right)]\exp\left[L(G^{k}(a),\omega^{2k-m-2}\lambda)+L(G^{-1}(a),\omega^{-2}\lambda)\right]\nonumber\\
&=[2\omega^{\frac{2-n}{2}+\mu(G^{-1}(a))}+O\left(\lambda^{-\rho}\right)]\exp\left[L(G^n(a),\omega^{2n}\lambda)-L(a,\lambda)\right]\nonumber\\
&+[2\omega^{\frac{2-n}{2}+\mu(G^{n-1}(a))}+O\left(\lambda^{-\rho}\right)]\exp\left[L(G^{-1}(a),\omega^{-2}\lambda)-L(G^{n-1}(a),\omega^{2(n-1)}\lambda)\right],\label{asy_sector}
\end{align}
where we used \eqref{1steq} for $\mathcal{W}_{0,n}(a,\lambda)$ and \eqref{2ndeq} for everything else, provided
that $\lambda$ lies in \eqref{dom_sector} and that
\begin{align}
\pi-\frac{4\left(\lfloor\frac{m}{2}\rfloor-n\right)\pi}{m+2}+\delta \leq &\arg(\omega^{2n}\lambda)\leq \pi-\delta\nonumber\\
\pi-\frac{4\lfloor\frac{m}{2}\rfloor\pi}{m+2}+\delta \leq &\arg(\lambda)\leq \pi-\delta\nonumber\\
\pi-\frac{4\left(\lfloor\frac{m}{2}\rfloor+1\right)\pi}{m+2}+\delta \leq &\arg(\omega^{-2}\lambda)\leq \pi-\delta,\nonumber
\end{align}
that is,
\begin{equation}\nonumber
-\frac{2n\pi}{m+2}+\delta\leq\arg(\lambda)\leq\pi -\frac{4n\pi}{m+2}-\delta.
\end{equation}
Thus, the second part of the theorem is proved by induction.

Next in order to  prove the first part of the theorem for the sector \eqref{zero1_sector}, we will determine which term in \eqref{asy_sector} dominates as $\lambda\to\infty$. To do that, we look at
\begin{align}
&\Re\left(L(G^{-1}(a),\omega^{-2}\lambda)-L(G^{n-1}(a),\omega^{2(n-1)}\lambda)\right)-\Re\left(L(G^{n}(a),\omega^{2n}\lambda)-L(a,\lambda)\right)\nonumber\\
&=K_m\left[\cos\left(-\frac{2\pi}{m}+\frac{m+2}{2m}\arg(\lambda)\right)-\cos\left(\frac{2(n-1)\pi}{m}+\frac{m+2}{2m}\arg(\lambda)\right)\right.\nonumber\\
&-\left.\left(\cos\left(-\frac{2n\pi}{m}+\frac{m+2}{2m}\arg(\lambda)\right)-\cos\left(\frac{m+2}{2m}\arg(\lambda)\right)\right)\right]|\lambda|^{\frac{1}{2}+\frac{1}{m}}(1+o(1))\nonumber\\
&=2K_m\sin\left(\frac{n\pi}{m}\right)\left[\sin\left(\frac{(n-2)\pi}{m}+\frac{m+2}{2m}\arg(\lambda)\right)\right.\nonumber\\
&\qquad\qquad\qquad\qquad\qquad\qquad\left.+\sin\left(\frac{n\pi}{m}+\frac{m+2}{2m}\arg(\lambda)\right)\right]|\lambda|^{\frac{1}{2}+\frac{1}{m}}(1+o(1))\nonumber\\
&=4K_m\sin\left(\frac{n\pi}{m}\right)\cos\left(\frac{\pi}{m}\right)\sin\left(\frac{(n-1)\pi}{m}+\frac{m+2}{2m}\arg(\lambda)\right)|\lambda|^{\frac{1}{2}+\frac{1}{m}}(1+o(1)),\label{dom_set}
\end{align}
that tends to positive infinity as $\lambda\to\infty$  (and hence the second term in \eqref{asy_sector} dominates) if $-\frac{2(n-1)\pi}{m+2}+\delta\leq\arg(\lambda)\leq \pi-\frac{4n\pi}{m+2}-\delta$.

We still need to prove \eqref{up_asy} for the sector
\begin{equation}\label{asy_rem}
\pi-\frac{4n\pi}{m+2}-\delta\leq\arg(\lambda)\leq \pi-\frac{4n\pi}{m+2}+\delta,
\end{equation}
for which we use induction on $n$ again.

When $n=1$, \eqref{up_asy} holds by Lemma \ref{lemma_up}.

Suppose that \eqref{up_asy} in the sector \eqref{asy_rem} for $n-1$ with  $2\leq n<\lfloor\frac{m}{2}\rfloor$.
Then \eqref{kplus1} and \eqref{kn_asy} with $k=n+1$  yield
\begin{equation}\nonumber
\mathcal{W}_{-1,n}(a,\lambda)
=\frac{2\omega^{\mu(G^{-1}(a))}\mathcal{W}_{0,1}(G^n(a),\omega^{2n}\lambda)}{\omega^{n-1}\mathcal{W}_{0,n+1}(a,\lambda)}+\frac{\mathcal{W}_{-1,n-1}(G(a),\omega^2\lambda)\mathcal{W}_{0,n+2}(G^{-1}(a),\omega^{-2}\lambda)}{\mathcal{W}_{0,n+1}(a,\lambda)}.
\end{equation}
If $\pi-\frac{4n\pi}{m+2}-\delta\leq\arg(\lambda)\leq \pi-\frac{4n\pi}{m+2}+\delta$, then $\pi-\frac{4(n-1)\pi}{m+2}-\delta\leq\arg(\omega^2\lambda)\leq \pi-\frac{4(n-1)\pi}{m+2}+\delta$. So
\begin{align}
\mathcal{W}_{-1,n}(a,\lambda)
&=\frac{4\omega^{\mu(G^{-1}(a))+\mu(G^n(a))}}{\omega^{n-1}\mathcal{W}_{0,n+1}(a,\lambda)}+\frac{\mathcal{W}_{-1,n-1}(G(a),\omega^2\lambda)\mathcal{W}_{0,n+2}(G^{-1}(a),\omega^{-2}\lambda)}{\mathcal{W}_{0,n+1}(a,\lambda)}\label{last_eq}\\
&=\frac{4\omega^{\mu(G^{-1}(a))+\mu(G^n(a))}}{\omega^{n-1}[2i\omega^{-\frac{n+1}{2}}+O\left(\lambda^{-\rho}\right)]\exp\left[L(G^{n+1}(a),\omega^{2n-m}\lambda)+L(a,\lambda)\right]}\nonumber\\
&+\frac{[2\omega^{\frac{3-n}{2}+\mu(G^{n-1}(a))}+O\left(\lambda^{-\rho}\right)]\exp\left[L(a,\lambda)-L(G^{n-1}(a),\omega^{2(n-1)}\lambda)\right]}{[2i\omega^{-\frac{n+1}{2}}+O\left(\lambda^{-\rho}\right)]\exp\left[L(G^{n+1}(a),\omega^{2n-m}\lambda)+L(a,\lambda)\right]}\nonumber\\
&\times [2i\omega^{-\frac{n+2}{2}}+O\left(\lambda^{-\rho}\right)]\exp\left[L(G^{n+1}(a),\omega^{2n-m}\lambda)+L(G^{-1}(a),\omega^{-2}\lambda)\right]\nonumber\\
&=[-2i\omega^{\frac{3-n}{2}+\mu(G^{-1}(a))+\mu(G^n(a))}+O\left(\lambda^{-\rho}\right)]\exp\left[-L(G^{n+1}(a),\omega^{2n-m}\lambda)-L(a,\lambda)\right]\nonumber\\
&+[2\omega^{\frac{2-n}{2}+\mu(G^{n-1}(a))}+O\left(\lambda^{-\rho}\right)]\exp\left[L(G^{-1}(a),\omega^{-2}\lambda)-L(G^{n-1}(a),\omega^{2(n-1)}\lambda)\right],\nonumber
\end{align}
where we use the induction hypothesis for $\mathcal{W}_{-1,n-1}(G(a),\omega^2\lambda)$ and use \eqref{sec_eq1} for
$\mathcal{W}_{0,n+1}(a,\lambda)$ and $\mathcal{W}_{0,n+2}(G^{-1}(a),\omega^{-2}\lambda)$. Next, we use an argument
similar \eqref{dom_set} to complete the induction step. Thus, the theorem is proved.
\end{proof}
Next we investigate $\mathcal{W}_{0,\lfloor\frac{m}{2}\rfloor}(a,\lambda)$.
\begin{theorem}\label{bod_thm}
If $m\geq 4$ is an even integer, then
\begin{align}
&\mathcal{W}_{-1,\lfloor\frac{m}{2}\rfloor}(a,\lambda)\nonumber\\
&=-[2\omega^{2+\mu(G^{-1}(a))+\mu(G^{\frac{m}{2}}(a))}+O\left(\lambda^{-\rho}\right)]\exp\left[-L(G^{{\frac{m}{2}+1}}(a),\lambda)-L(a,\lambda)\right]\label{thm_eq1}\\
&-[2\omega^{2+\mu(a)+\mu(G^{\frac{m-2}{2}}(a))}+O\left(\lambda^{-\rho}\right)]\exp\left[-L(G^{\frac{m-2}{2}}(a),\omega^{m-2}\lambda)-L(G^{m}(a),\omega^{m-2}\lambda)\right].\nonumber
\end{align}
as $\lambda\to\infty$ in the sector
\begin{equation}\label{zero4_sector}
-\pi+\frac{4\pi}{m+2}-\delta\leq\arg(\lambda)\leq-\pi+\frac{4\pi}{m+2}+\delta.
\end{equation}

If $m\geq 4$ is an odd integer, then
\begin{align}
\mathcal{W}_{-1,\lfloor\frac{m}{2}\rfloor}(a,\lambda)
&=[2\omega^{\frac{5}{4}}+O\left(\lambda^{-\rho}\right)]\exp\left[-L(G^{\frac{m+1}{2}}(a),\omega^{-1}\lambda)-L(a,\lambda)\right]\nonumber\\
&+[2\omega^{\frac{5}{4}}+O\left(\lambda^{-\rho}\right)]\exp\left[L(G^{m+1}(a),\omega^{-2}\lambda)-L(G^{\frac{m-3}{2}}(a),\omega^{m-3}\lambda)\right].\label{thm_eq2}
\end{align}
as $\lambda\to\infty$ in the sector
\begin{equation}\nonumber
-\pi+\frac{4\pi}{m+2}+\delta\leq\arg(\lambda)\leq-\pi+\frac{6\pi}{m+2}+\delta.
\end{equation}
\end{theorem}
\begin{proof}
We will use \eqref{last_eq} with $n=\lfloor\frac{m}{2}\rfloor$, that is,
\begin{equation}\nonumber
\mathcal{W}_{-1,\lfloor\frac{m}{2}\rfloor}(a,\lambda)=\frac{4\omega^{\mu(G^{-1}(a))+\mu(G^{\lfloor\frac{m}{2}\rfloor}(a))}}{\omega^{\lfloor\frac{m}{2}\rfloor-1}\mathcal{W}_{0,\lfloor\frac{m}{2}\rfloor+1}(a,\lambda)}+\frac{\mathcal{W}_{-1,\lfloor\frac{m}{2}\rfloor-1}(G(a),\omega^2\lambda)\mathcal{W}_{0,\lfloor\frac{m}{2}\rfloor+2}(G^{-1}(a),\omega^{-2}\lambda)}{\mathcal{W}_{0,\lfloor\frac{m}{2}\rfloor+1}(a,\lambda)}.
\end{equation}
When $m$ is even, say $m=2k$,
\begin{align}
\mathcal{W}_{0,k+2}(G^{-1}(a),\omega^{-2}\lambda)&=\mathcal{W}_{m+2,k+2}(G^{-1}(a),\omega^{-2}\lambda)\nonumber\\
&=-\omega^{-k-3}\mathcal{W}_{-1,k-1}(G^{k+2}(a),\omega^{2k+4}\lambda)\nonumber\\
&=\omega^{-2}\mathcal{W}_{-1,k-1}(G^{k+2}(a),\omega^{2}\lambda).\nonumber
\end{align}
So
\begin{align}
\mathcal{W}_{-1,k}(a,\lambda)&=\frac{4\omega^{\mu(G^{-1}(a))+\mu(G^{k}(a))}}{\omega^{k-1}\mathcal{W}_{0,k+1}(a,\lambda)}+\frac{\mathcal{W}_{-1,k-1}(G(a),\omega^2\lambda)\mathcal{W}_{0,k+2}(G^{-1}(a),\omega^{-2}\lambda)}{\mathcal{W}_{0,k+1}(a,\lambda)}\nonumber\\
&=\frac{4\omega^{\mu(G^{-1}(a))+\mu(G^{k}(a))}}{\omega^{k-1}\mathcal{W}_{0,k+1}(a,\lambda)}+\frac{\mathcal{W}_{-1,k-1}(G(a),\omega^2\lambda)\mathcal{W}_{-1,k-1}(G^{k+2}(a),\omega^{2}\lambda)}{\omega^2
\mathcal{W}_{0,k+1}(a,\lambda)}\nonumber
\end{align}
Since $\lambda$ lies in \eqref{zero4_sector},
$$
-\frac{2\left(\lfloor\frac{m}{2}\rfloor-2\right)\pi}{m+2}+\delta\leq-\pi+\frac{8\pi}{m+2}-\delta\leq\arg(\omega^2\lambda)\leq\pi-\frac{4\left(\lfloor\frac{m}{2}\rfloor-1\right)\pi}{m+2}+\delta.
$$
\begin{align}
&\mathcal{W}_{-1,k}(a,\lambda)\nonumber\\
&=\frac{4\omega^{\mu(G^{-1}(a))+\mu(G^{k}(a))}}{\omega^{k-1}\mathcal{W}_{0,k+1}(a,\lambda)}+\frac{\mathcal{W}_{-1,k-1}(G(a),\omega^2\lambda)\mathcal{W}_{0,k+2}(G^{-1}(a),\omega^{-2}\lambda)}{\mathcal{W}_{0,k+1}(a,\lambda)}\nonumber\\
&=\frac{4\omega^{\mu(G^{-1}(a))+\mu(G^{k}(a))}}{\omega^{k-1}\mathcal{W}_{0,k+1}(a,\lambda)}+\frac{\mathcal{W}_{-1,k-1}(G(a),\omega^2\lambda)\mathcal{W}_{-1,k-1}(G^{k+2}(a),\omega^{2}\lambda)}{\omega^2 \mathcal{W}_{0,k+1}(a,\lambda)}\nonumber\\
&=\frac{4\omega^{\mu(G^{-1}(a))+\mu(G^{k}(a))}}{\omega^{k-1}[2i\omega^{-\frac{k+1}{2}}+O\left(\lambda^{-\rho}\right)]\exp\left[L(G^{k+1}(a),\omega^{2k-m}\lambda)+L(a,\lambda)\right]}\nonumber\\
&+\frac{[2\omega^{\frac{3-k}{2}+\mu(G^{k-1}(a))}+O\left(\lambda^{-\rho}\right)]\exp\left[L(a,\lambda)-L(G^{k-1}(a),\omega^{2(k-1)}\lambda)\right]}{\omega^2 [2i\omega^{-\frac{k+1}{2}}+O\left(\lambda^{-\rho}\right)]\exp\left[L(G^{k+1}(a),\omega^{2k-m}\lambda)+L(a,\lambda)\right]}\nonumber\\
&\times[2\omega^{\frac{3-k}{2}+\mu(G^{2k}(a))}+O\left(\lambda^{-\rho}\right)]\exp\left[L(G^{k+1}(a),\lambda)-L(G^{2k}(a),\omega^{2(k-1)}\lambda)\right]\nonumber\\
&=[-2i\omega^{\frac{3-k}{2}+\mu(G^{-1}(a))+\mu(G^{k}(a))}+O\left(\lambda^{-\rho}\right)]\exp\left[-L(G^{k+1}(a),\lambda)-L(a,\lambda)\right]\nonumber\\
&-[2i\omega^{\frac{3-k}{2}+\mu(G^{k-1}(a))+\mu(G^{2k}(a))}+O\left(\lambda^{-\rho}\right)]\exp\left[-L(G^{k-1}(a),\omega^{2(k-1)}\lambda)-L(G^{2k}(a),\omega^{2(k-1)}\lambda)\right]\nonumber\\
&=-[2i\omega^{\frac{3-k}{2}+\mu(G^{-1}(a))+\mu(G^{k}(a))}+O\left(\lambda^{-\rho}\right)]\exp\left[-L(G^{k+1}(a),\lambda)-L(a,\lambda)\right]\nonumber\\
&-[2i\omega^{\frac{3-k}{2}+\mu(a)+\mu(G^{k-1}(a))}+O\left(\lambda^{-\rho}\right)]\exp\left[-L(G^{k-1}(a),\omega^{2(k-1)}\lambda)-L(G^{2k}(a),\omega^{2(k-1)}\lambda)\right],\nonumber
\end{align}
where we used \eqref{up_asy} for $\mathcal{W}_{-1,k-1}(G(a),\cdot)$ and $\mathcal{W}_{-1,k-1}(G^{k+2}(a),\cdot)$,
and \eqref{sec_eq1} for $\mathcal{W}_{0,k+1}(a,\cdot)$. Finally, we use $\omega^{-\frac{m+2}{4}}=-i$, to get the
desired asymptotic expansion of $\mathcal{W}_{-1,\lfloor\frac{m}{2}\rfloor}(a,\lambda)$.

Next we investigate the case when $m$ is odd, say $m=2k+1$ (so $\lfloor\frac{m}{2}\rfloor=k$).
\begin{equation}\nonumber
\mathcal{W}_{0,\lfloor\frac{m}{2}\rfloor+1}(a,\lambda)=\omega^{-1}\mathcal{W}_{-1,\lfloor\frac{m}{2}\rfloor}(G(a),\omega^{2}\lambda)
\end{equation}
and
\begin{align}
\mathcal{W}_{0,k+2}(G^{-1}(a),\omega^{-2}\lambda)&=\mathcal{W}_{m+2,k+2}(G^{-1}(a),\omega^{-2}\lambda)\nonumber\\
&=-\omega^{-k-2}\mathcal{W}_{0,k+1}(G^{k+1}(a),\omega^{2k+2}\lambda)\nonumber\\
&=\omega^{-\frac{1}{2}}\mathcal{W}_{0,k+1}(G^{k+1}(a),\omega^{-1}\lambda).\nonumber
\end{align}
Similarly to the proof of the theorem for $m$ even,
\begin{align}
\mathcal{W}_{-1,k}(a,\lambda)
&=\frac{4\omega^{\mu(G^{-1}(a))+\mu(G^{k}(a))}}{\omega^{k-1}\mathcal{W}_{0,k+1}(a,\lambda)}+\frac{\mathcal{W}_{-1,k-1}(G(a),\omega^2\lambda)\mathcal{W}_{0,k+2}(G^{-1}(a),\omega^{-2}\lambda)}{\mathcal{W}_{0,k+1}(a,\lambda)}\nonumber\\
&=\frac{4\omega^{\mu(G^{-1}(a))+\mu(G^{k}(a))}}{\omega^{k-1}\mathcal{W}_{0,k+1}(a,\lambda)}+\frac{\mathcal{W}_{-1,k-1}(G(a),\omega^2\lambda)\mathcal{W}_{0,k+1}(G^{k+1}(a),\omega^{-1}\lambda)}{\omega^{\frac{1}{2}} \mathcal{W}_{0,k+1}(a,\lambda)}\nonumber
\end{align}
\begin{align}
\hskip 70 pt&=\frac{4\omega^{\mu(G^{-1}(a))+\mu(G^{k}(a))}}{\omega^{k-1}[2i\omega^{-\frac{k+1}{2}}+O\left(\lambda^{-\rho}\right)]\exp\left[L(G^{k+1}(a),\omega^{2k-m}\lambda)+L(a,\lambda)\right]}\nonumber\\
&+\frac{[2\omega^{\frac{3-k}{2}+\mu(G^{k-1}(a))}+O\left(\lambda^{-\rho}\right)]\exp\left[L(a,\lambda)-L(G^{k-1}(a),\omega^{2(k-1)}\lambda)\right]}{\omega^{\frac{1}{2}} [2i\omega^{-\frac{k+1}{2}}+O\left(\lambda^{-\rho}\right)]\exp\left[L(G^{k+1}(a),\omega^{2k-m}\lambda)+L(a,\lambda)\right]}\nonumber\\
&\times[2i\omega^{-\frac{k+1}{2}}+O\left(\lambda^{-\rho}\right)]\exp\left[L(G^{2k+2}(a),\omega^{2k-m-1}\lambda)+L(G^{k+1}(a),\omega^{-1}\lambda)\right]\nonumber\\
&=[-2i\omega^{\frac{3-k}{2}+\mu(G^{-1}(a))+\mu(G^{k}(a))}+O\left(\lambda^{-\rho}\right)]\exp\left[-L(G^{k+1}(a),\omega^{-1}\lambda)-L(a,\lambda)\right]\nonumber\\
&+[2\omega^{\frac{2-k}{2}+\mu(G^{k-1}(a))}+O\left(\lambda^{-\rho}\right)]\exp\left[L(G^{2k+2}(a),\omega^{-2}\lambda)-L(G^{k-1}(a),\omega^{2(k-1)}\lambda)\right]\nonumber\\
&=-[2i\omega^{\frac{3-k}{2}+\frac{m}{2}}+O\left(\lambda^{-\rho}\right)]\exp\left[-L(G^{k+1}(a),\omega^{-1}\lambda)-L(a,\lambda)\right]\nonumber\\
&+[2\omega^{\frac{2-k}{2}+\frac{m}{4}}+O\left(\lambda^{-\rho}\right)]\exp\left[L(G^{2k+2}(a),\omega^{-2}\lambda)-L(G^{k-1}(a),\omega^{2(k-1)}\lambda)\right],\nonumber
\end{align}
where we use \eqref{up_asy} for $\mathcal{W}_{-1,k-1}(G(a),\cdot)$, and use \eqref{sec_eq1} for
$\mathcal{W}_{0,k+1}(a,\cdot)$ and  $\mathcal{W}_{0,k+1}(G^{k+1}(a),\cdot)$. Finally, we use
$\omega^{\frac{m+2}{4}}=i$, to get the asymptotic expansion of
$\mathcal{W}_{-1,\lfloor\frac{m}{2}\rfloor}(a,\lambda)$. This completes the proof.
\end{proof}

 The {\it order of an entire function} $g$ is defined by
$$\limsup_{r\rightarrow \infty}\frac{\log \log M(r,g)}{\log r},$$
where $M(r, g)=\max \{|g(re^{i\theta})|: 0\leq \theta\leq 2\pi\}$ for $r>0$.
If for some positive real numbers $\sigma,\, c_1,\, c_2$, we have $\exp[c_1 r^{\sigma}]\leq M(r,g)\leq  \exp[c_2 r^{\sigma}]$ for all large $r$, then the order of $g$ is $\sigma$.
\begin{corollary}\label{zero_free}
Let $1\leq n\leq \frac{m}{2}$. Then
 the entire functions $\mathcal{W}_{-1,n}(a,\cdot)$  are of order $\frac{1}{2}+\frac{1}{m}$, and hence they have infinitely many zeros in the complex plane. Moreover,  $\mathcal{W}_{-1,n}(a,\cdot)$ have at most finitely many zeros in the sector
$$
-\frac{2(n-1)\pi}{m+2}+\delta\leq\arg(\lambda)\leq\pi+\frac{4\pi}{m+2}-\delta.
$$
\end{corollary}

\section{Relation between eigenvalues of $H_{\ell}$ and zeros of $\mathcal{W}_{-1,n}(a,\cdot)$}\label{asymp_eigen}
In this section, we will relate the eigenvalues of $H_{\ell}$ with zeros of some entire function
$\mathcal{W}_{-1,n}(a,\cdot)$.

Suppose that $\ell=2s-1$ is odd with $1\leq \ell\leq m-1$. Then \eqref{ptsym} becomes \eqref{rotated} by the
scaling $v(z)=u(-iz)$, and $v$ decays in the Stokes sectors $S_{-s}$ and $S_{s}$. Since $f_{s-1}$ and $f_{s}$ are
linearly independent, for some $D_s$ and $\widetilde{D}_s$ one can write
$$f_{-s}(z,a,\lambda)=D_s(a,\lambda)f_{s-1}(z,a,\lambda)+\widetilde{D}_s(a,\lambda)f_{s}(z,a,\lambda).$$
Then one finds
$$
D_s(a,\lambda)=\frac{\mathcal{W}_{-s,s}(a,\lambda)}{\mathcal{W}_{s-1,s}(a,\lambda)}\quad\text{and}\quad
\widetilde{D}_s(a,\lambda)=\frac{\mathcal{W}_{-s,s-1}(a,\lambda)}{\mathcal{W}_{s,s-1}(a,\lambda)}.
$$
Also it is easy to see that $\lambda$ is an eigenvalue of $H_{\ell}$ if and only if $D_s(a,\lambda)=0$ if and only
if $\mathcal{W}_{-s,s}(a,\lambda)=0$. Since
$\mathcal{W}_{-s,s}(a,\lambda)=\omega^{s-1}\mathcal{W}_{-1,2s-1}(G^{-s+1}(a),\omega^{-2s+2}\lambda)$, by Corollary
\ref{zero_free} $\mathcal{W}_{-s,s}(a,\lambda)$ has at most finitely many zeros in the sector $
-\frac{2(2s-2)\pi}{m+2}+\delta\leq\arg(\omega^{-2s+2}\lambda)\leq\pi+\frac{4\pi}{m+2}-\delta, $ that is,
$$
\delta\leq\arg(\lambda)\leq\pi+\frac{2(2s+1)\pi}{m+2}-\delta.
$$

Next, by symmetry one can show that $\mathcal{W}_{-s,s}(a,\lambda)$ has at most finitely many zeros in the sector
$\pi\leq\arg(\lambda)\leq 2\pi-\delta$. For that, we examine $H_{\ell}$ with $P(z)$ replaced by
$\overline{P(\overline{z})}$ whose coefficient vector is
$\overline{a}:=(\overline{a}_1,\,\overline{a}_2,\dots,\overline{a}_{m})$. Then one sees that
$\mathcal{W}_{-s,s}(a,\lambda)=0$ if and only if $\mathcal{W}_{-s,s}(\overline{a},\overline{\lambda})=0$.
$\mathcal{W}_{-s,s}(\overline{a},\overline{\lambda})$ has at most finitely many zeros in the sector
$\delta\leq\arg(\overline{\lambda})\leq\pi$ by the arguments above.  Thus, $\mathcal{W}_{-s,s}(a,\lambda)$ has at
most finitely many zeros in the sector $\pi\leq\arg(\lambda)\leq 2\pi-\delta$, and has infinitely many zeros in
the sector $|\arg(\lambda)|\leq\delta$ since it is an entire function of order $\frac{1}{2}+\frac{1}{m}\in(0,1)$.

Suppose that $\ell=2s$ is even with $1\leq \ell\leq m-1$. Then \eqref{ptsym} becomes \eqref{rotated2} by the
scaling $y(z)=u(-i\omega^{-\frac{1}{2}}z)$, and $y$ decays in the Stokes sectors $S_{-s}$ and $S_{s+1}$. We then
see that the coefficient vector $\widetilde{a}$ of the polynomial $\omega^{-1}P(\omega^{-\frac{1}{2}}z)$ becomes
$$\widetilde{a}=G^{-\frac{1}{2}}(a).$$

Now one can express $f_{-s}$ as a linear combination of
 $f_{s}$ and $f_{s+1}$ as follows.
$$f_{-s}(z,\widetilde{a} ,\omega^{-1}\lambda)=\frac{\mathcal{W}_{-s,s+1}(\widetilde{a},\omega^{-1}\lambda)}{\mathcal{W}_{s,s+1}(\widetilde{a},\omega^{-1}\lambda)}f_{s}(z,\widetilde{a},\omega^{-1}\lambda)+\frac{\mathcal{W}_{-s,s}(\widetilde{a},\omega^{-1}\lambda)}{\mathcal{W}_{s+1,s}(\widetilde{a},\omega^{-1}\lambda)}f_{s+1}(z,\widetilde{a},\omega^{-1}\lambda).$$
Thus,  $\lambda$ is an eigenvalue of $H_{\ell}$ if and only if
$\mathcal{W}_{-s,s+1}(\widetilde{a},\omega^{-1}\lambda)=0$. Since
$$\mathcal{W}_{-s,s+1}(\widetilde{a},\omega^{-1}\lambda)=\omega^{s-1}\mathcal{W}_{-1,2s}(G^{-s+1}(\widetilde{a}),\omega^{-2s+1}\lambda),$$
  by Corollary \ref{zero_free},
 $\mathcal{W}_{-s,s+1}(\widetilde{a},\omega^{-1}\lambda)$ has at most finitely many zeros in the sector
$ -\frac{2(2s-1)\pi}{m+2}+\delta\leq\arg(\omega^{-2s+1}\lambda)\leq\pi+\frac{4\pi}{m+2}-\delta, $ that is,
$$
\delta\leq\arg(\lambda)\leq\pi+\frac{4s\pi}{m+2}-\delta.
$$
This is true for each $a\in\C^{m}$. So one can show that $\mathcal{W}_{-s,s+1}(\widetilde{a},\omega^{-1}\lambda)$
has at most finitely many zeros in the sector $\pi\leq\arg(\lambda)\leq 2\pi-\delta$ by symmetry, similar to the
case when $\ell$ is odd.   Thus, $\mathcal{W}_{-s,s+1}(a,\lambda)$  has infinitely many zeros in the sector
$|\arg(\lambda)|\leq\delta$.


\section{Proof of Theorem \ref{main_thm1} when $1\leq\ell<\lfloor\frac{m}{2}\rfloor$}\label{sec_7}

In this section, we prove Theorem \ref{main_thm1} for
$1\leq\ell<\lfloor\frac{m}{2}\rfloor$ and in doing so, we will use
the following proposition on univalent functions.
\begin{proposition}[\protect{\cite[p.\ 216]{TK}}]\label{gen_pro1}
${} $ Let $A(\mu)$ be  analytic in the region $S=\{\mu\in\C:
\alpha\leq \arg(\mu)\leq\beta, |\mu|\geq M\}$ for some
$\alpha,\beta\in\R$ with $\beta-\alpha<\pi$ and for some $M>0$.
Suppose that $A(\mu)\to0$ as $\mu\to\infty$ in $S$. Then there
exist $\widetilde{\alpha}, \widetilde{\beta}, \widetilde{M}$ such
that $\mu(1+A(\mu))$ is univalent in $\{\mu\in\C:
\widetilde{\alpha}\leq \arg(\mu)\leq\widetilde{\beta}, |\mu|\geq
\widetilde{M}\}\subset S$.
\end{proposition}
\begin{proof}
See \cite[p.\ 216]{TK} or \cite[Thm.\ 3.4]{Sibuya} for a proof.
\end{proof}

We will consider two cases; when $\ell$ is odd and when $\ell$ is
even.
\begin{proof}[Proof of Theorem ~\ref{main_thm1} when $1\leq\ell<\lfloor\frac{m}{2}\rfloor$ is odd]
We will closely follow his proof of Theorem 29.1 in \cite{Sibuya}
where Sibuya  computed the leading term in the asymptotics
\eqref{main_result} for  $\ell=1$.

Suppose that $1\leq \ell=2s-1<\lfloor\frac{m}{2}\rfloor$ is odd.
Recall that when $\ell$ is odd, $\lambda$ is an eigenvalue of
$H_{\ell}$ if and only if
 $\mathcal{W}_{-s,s}(a,\lambda)=0$. Also, in Section \ref{asymp_eigen} we showed that $\mathcal{W}_{-s,s}(a,\lambda)$ has all  zeros except finitely many in the sector
 $|\arg(\lambda)|\leq\delta$ and hence,
all the eigenvalues $\lambda$ of  $H_{\ell}$ lie in the sector $|\arg(\lambda)|\leq\delta$ if $|\lambda|$ is large
enough.

Since
\begin{equation}\nonumber
\mathcal{W}_{-s,s}(a,\lambda)=\omega^{s-1}\mathcal{W}_{-1,2s-1}(G^{-s+1}(a),\omega^{-2s+2}\lambda),
\end{equation}
 we will use \eqref{zero_asy} to investigate asymptotics of large eigenvalues. Suppose that $\mathcal{W}_{-s,s}(a,\lambda)=0$ and $|\lambda|$ is large enough.
 Then from \eqref{zero_asy} with $n=2s-1$, and with $a$ and $\lambda$ replaced by $G^{-s+1}(a)$ and $\omega^{-2s+2}\lambda$, respectively, we have
\begin{align}
\left[1+O\left(\lambda^{-\rho}\right)\right]&\exp\left[L(G^{s}(a),\omega^{2s}\lambda)-L(G^{-s}(a),\omega^{-2s}\lambda)\right]\nonumber\\
&\times\exp\left[L(G^{s-1}(a),\omega^{2s-2}\lambda)-L(G^{-s+1}(a),\omega^{-2s+2}\lambda)\right]=-\omega^{2\nu(G^s(a))}.\nonumber
\end{align}
Also, since $\left[1+O\left(\lambda^{-\rho}\right)\right]=\exp\left[O\left(\lambda^{-\rho}\right)\right]$ and
$\omega^{2\nu(G^s(a))}=\exp\left[\frac{4\pi \nu(G^s(a))}{m+2}\right]$,
\begin{align}
&\exp\left[L(G^{s}(a),\omega^{2s}\lambda)-L(G^{-s}(a),\omega^{-2s}\lambda)\right]\nonumber\\
&\times\exp\left[L(G^{s-1}(a),\omega^{2s-2}\lambda)-L(G^{-s+1}(a),\omega^{-2s+2}\lambda)-\frac{4\pi
\nu(G^s(a))}{m+2}+O\left(\lambda^{-\rho}\right)\right]=-1.\label{sim_eq1}
\end{align}

For each odd integer $\ell=2s-1$ in $1\leq\ell<\lfloor\frac{m}{2}\rfloor$, we define
\begin{align}
h_{m,\ell}(\lambda)=& L(G^{s}(a),\omega^{2s}\lambda)-L(G^{-s}(a),\omega^{-2s}\lambda)\nonumber\\
&+L(G^{s-1}(a),\omega^{2s-2}\lambda)-L(G^{-s+1}(a),\omega^{-2s+2}\lambda)-\frac{4\pi
\nu(G^s(a))}{m+2}+O\left(\lambda^{-\rho}\right).\nonumber
\end{align}
Then by Corollary \ref{lemma_decay},
\begin{align}
h_{m,\ell}(\lambda)&=K_m\left(e^{\frac{2s\pi}{m}i}-e^{-\frac{2s\pi}{m}i}+e^{\frac{2(s-1)\pi}{m}i}-e^{-\frac{2(s-1)\pi}{m}i}\right)\lambda^{\frac{m+2}{2m}}(1+o(1))\nonumber\\
&=2iK_m\left(\sin\left(\frac{2s\pi}{m}\right)+\sin\left(\frac{2s\pi}{m}-\frac{2\pi}{m}\right)\right)\lambda^{\frac{m+2}{2m}}(1+o(1))\nonumber\\
&=4iK_m\cos\left(\frac{\pi}{m}\right)\sin\left(\frac{(2s-1)\pi}{m}\right)\lambda^{\frac{m+2}{2m}}(1+o(1))\quad\text{as}\quad\lambda\to\infty,\label{assco_eq}
\end{align}
in the sector $|\arg(\lambda)|\leq\delta$. Since $K_m>0$ and $0<\frac{(2s-1)\pi}{m}<\pi$, the function
$h_{m,\ell}(\cdot)$ maps the region $|\lambda|\geq M_1$ for some large $M_1$ and $|\arg(\lambda)|\leq\delta$ into
a region containing $|\lambda|\geq M_2$ for some large $M_2$ and $|\arg(\lambda)-\frac{\pi}{2}|\leq\varepsilon_1$
for some $\varepsilon_1>0$. Also, $h_{m,\ell}(\cdot)$ is analytic in this region due to the analyticity of
$f(0,a,\lambda)$ in Theorem \ref{prop}. Following Sibuya, we will show that for every large positive integer $n$
there exists $\lambda_n$ such that
\begin{equation}\label{h_eq}
h_{m,\ell}(\lambda_{n}) =\left(2n+1\right)\pi i.
\end{equation}

Next, Proposition \ref{gen_pro1} with $\mu=\lambda^{\frac{m+2}{2m}}$ implies that there exist $M_1^\d>0$ and
$0<\delta^\d<\frac{\pi}{6}$ such that $\mu(1+A(\mu))$ is univalent in the sector $|\arg(\mu)|\leq\delta^\d$ and
$|\mu|\geq M_1^\d.$ Thus, for each $n\in\Z$, there is at most one $\mu_n$ in the sector $|\arg(\mu)|\leq\delta^\d$
and $|\mu|\geq M_1^\d$ such that
\begin{equation}\label{B_eq}
\mu_n(1+A(\mu_n))=B_n,
\end{equation}
where $A$ is the error term in \eqref{assco_eq} and
\begin{equation}\nonumber
B_n=\frac{\left(2n+1\right)\pi}{4K_m\cos\left(\frac{\pi}{m}\right)\sin\left(\frac{(2s-1)\pi}{m}\right)}.
\end{equation}
Then there exists $M_1^\dd\geq M_1^\d$ such that in  the sector $|\arg(\mu)|\leq\delta^\d$ and $|\mu|\geq
M_1^\dd$,
\begin{equation}\nonumber
|A(\mu)|<\frac{1}{5}.
\end{equation}
Also we can take $n$ large enough so that $\frac{1}{2}B_n\geq M_1^\dd$. Next, we define
\begin{equation}\nonumber
A_n:=\{|A(\mu)|: |\arg(\mu)|\leq\delta^\d, |\mu|\geq B_n\}.
\end{equation}
Then $\lim_{n\to\infty}A_n=0$

Suppose that $n$ is so large that $A_n\leq \frac{1}{2}\sin(\delta^\d)<\frac{1}{4}$. Then the disk defined by
$|\mu-B_n|\leq \frac{A_n}{1-2A_n}B_n$ is contained in the sector $|\arg(\mu)|\leq\delta^\d$ and
$|\mu|\geq\frac{1}{2}B_n$. Moreover, on the circle $|\mu-B_n|=\frac{A_n}{1-2A_n}B_n$,
\begin{align}
|\mu-B_n|-|\mu||A(\mu)|&\geq|\mu-B_n|-|\mu|A_n\nonumber\\
&\geq |\mu-B_n|(1-A_n)-A_nB_n\nonumber\\
&\geq \frac{1-A_n}{1-2A_n}A_nB_n-A_nB_n>0.\nonumber
\end{align}
Thus, by the Rouch\'e's theorem in complex analysis (see, e.\ g., \cite[p. 125]{Conway}), $\mu-B_n$ and
$\mu(1+A(\mu))-B_n$ have the same number of zeros in the disk and hence, \eqref{B_eq} has exactly one $\mu_n$ in
the sector. Therefore, there exists $N=N(m,a)>0$ such that \eqref{h_eq} has exactly one solution $\lambda_n$ for
all integers $n\geq N$.

Next, since
\begin{align}
&-\frac{\nu(G^{s}(a))}{m}\ln(\omega^{2s}\lambda)+\frac{\nu(G^{-s}(a))}{m}\ln(\omega^{-2s}\lambda)\nonumber\\
&-\frac{\nu(G^{s-1}(a))}{m}\ln(\omega^{2s-2}\lambda)+\frac{\nu(G^{-s+1}(a))}{m}\ln(\omega^{-2s+2}\lambda)-\frac{4\nu(G^s(a))}{m+2}\pi i\nonumber\\
=& \frac{\nu(G^{s}(a))}{m}\left(-\ln(\omega^{2s}\lambda)+\ln(\omega^{-2s}\lambda)
+\ln(\omega^{2s-2}\lambda)-\ln(\omega^{-2s+2}\lambda)\right)-\frac{4\nu(G^s(a))}{m+2}\pi i\nonumber\\
=& -\frac{4}{m}\nu(G^{s}(a))\pi i,\nonumber
\end{align}
using Corollary \ref{asy_lemma}, \eqref{h_eq} becomes
\begin{equation}\label{const_term}
\left(2n+1\right)\pi
i=\sum_{j=0}^{m+1}e_{\ell,j}(a)\lambda_{n}^{\frac{1}{2}+\frac{1-j}{m}}-\frac{4}{m}\nu(G^{s}(a))\pi
i+O\left(\lambda^{-\rho}\right),
\end{equation}
where
\begin{align}
e_{\ell,j}(a)=&K_{m,j}(G^{s}(a))\omega^{2s\left(\frac{1}{2}+\frac{1-j}{m}\right)}-K_{m,j}(G^{-s}(a))\omega^{-2s\left(\frac{1}{2}+\frac{1-j}{m}\right)}\nonumber\\
&+K_{m,j}(G^{s-1}(a))\omega^{2(s-1)\left(\frac{1}{2}+\frac{1-j}{m}\right)}-K_{m,j}(G^{-s+1}(a))\omega^{-2(s-1)\left(\frac{1}{2}+\frac{1-j}{m}\right)}\nonumber\\
=&K_{m,j}(a)\left(\omega^{-sj+2s\left(\frac{1}{2}+\frac{1-j}{m}\right)}-\omega^{sj-2s\left(\frac{1}{2}+\frac{1-j}{m}\right)}\right)\nonumber\\
&+K_{m,j}(a)\left(\omega^{-(s-1)j+2(s-1)\left(\frac{1}{2}+\frac{1-j}{m}\right)}-\omega^{(s-1)j-2(s-1)\left(\frac{1}{2}+\frac{1-j}{m}\right)}\right)\nonumber\\
=&4iK_{m,j}(a)\sin\left(\frac{(1-j)\ell\pi}{m}\right)\cos\left(\frac{(1-j)\pi}{m}\right),\nonumber
\end{align}
where we used Lemma \ref{lemma_25}.

In summary, we have showed that for each $a\in\C^m$, there exists
$N_0=N_0(m,a)\leq N$ such that $\{\lambda_n\}_{n\geq N_0}$ is the
set of all eigenvalues that satisfy \eqref{const_term},  since
there are only at most finitely many eigenvalues outside
$|\arg(\lambda)|\leq \varepsilon$ or inside $|\lambda|\leq M$ for
any $\varepsilon>0,\,M>0$. In order to complete proof, we need to
show that $N_0$ does not depend on $a\in\C^m$.

We will prove that for each $R>0$, if $|a|\leq R$, then $N_0(m,a)=N_0(m,0)$.
 The error terms in \eqref{eq1} and \eqref{eq2} are uniform on the closed ball $|a|\leq R$ and hence, so is $A(\mu)$ as $\mu\to\infty$
 in the sector $|\arg(\mu)|\leq \delta^\d$. Thus, we can choose $N$ independent of $|a|\leq R$. Since $\mu(1+A(\mu))$ is univalent
 in the sector $|\arg(\mu)|\leq\delta^\d$ and $|\mu|\geq M_1^\dd\geq M_1^\d$, by the implicit function theorem and \eqref{B_eq} the function $a\mapsto\mu_n(a)$ is
 continuous on $|a|\leq R$ for each $n\geq N$ and hence, so is $a\mapsto\lambda_n(a)$ for $n\geq N$.

 Next, we claim that for all $|a|\leq R$, there are exactly the same number of
  eigenvalues that are not in $\{\lambda_n(a)\}_{n\geq N}$. To prove this, we will use the Hurwitz's theorem  (see, e. g., \cite[p. 152]{Conway}) in complex analysis. That is,
   if a sequence of analytic functions converges uniformly to an
analytic function on any compact sets, then  eventually functions in the sequence and the limit function have the
same number of zeros in any open set whose boundary does not contain any zeros of the limit function. The
Hurwitz's theorem implies that since the eigenvalues are the zeros of the entire function
$\mathcal{W}_{-s,s}(a,\lambda)$, they vary continuously as $a$ and hence, there is no sudden appearance or
disappearance of eigenvalues. Also, none of the eigenvalues that are not in $\{\lambda_n(a)\}_{n\geq N}$ can be
continuously mapped to $\lambda_n(a)$ for some $n\geq N$ as $a$ varies and hence, the claim is proved. This
completes the proof.
\end{proof}

Next we prove Theorem ~\ref{main_thm1} for $1\leq\ell<\lfloor\frac{m}{2}\rfloor$ is even.
\begin{proof}[Proof of Theorem ~\ref{main_thm1} when $1\leq\ell<\lfloor\frac{m}{2}\rfloor$ is even]

Proof is very similar to the case when $1\leq\ell<\lfloor\frac{m}{2}\rfloor$ is odd. Let $\ell=2s$ for some
$s\in\N$.

Recall that  $\lambda$ is an eigenvalue of $H_{\ell}$ if and only if
$\mathcal{W}_{-1,2s}(G^{-s+1}(\widetilde{a}),\omega^{-2s+1}\lambda)=0$. Then from \eqref{zero_asy},  we have
\begin{align}
&\exp\left[L(G^{s+1}(\widetilde{a}),\omega^{2s+1}\lambda)-L(G^{-s}(\widetilde{a}),\omega^{-2s-1}\lambda)\right]\nonumber\\
&\times\exp\left[L(G^{s}(\widetilde{a}),\omega^{2s-1}\lambda)-L(G^{-s+1}(\widetilde{a}),\omega^{-2s+1}\lambda)+O\left(\lambda^{-\rho}\right)\right]=-1,\label{sim_eq2}
\end{align}
where $\widetilde{a}=G^{-\frac{1}{2}}(a)$ and where we used
$\left[1+O\left(\lambda^{-\rho}\right)\right]=\exp[O\left(\lambda^{-\rho}\right)]$
again.

Like in the case when $1\leq\ell<\lfloor\frac{m}{2}\rfloor$ is
odd, from Lemma \ref{asy_lemma}, we have
\begin{equation}\label{asy_eq_even}
\left(2n+1\right)\pi
i=\sum_{j=0}^{m+1}e_{\ell,j}(a)\lambda_{n}^{\frac{1}{2}+\frac{1-j}{m}}+O\left(\lambda^{-\rho}\right),
\end{equation}
where for $0\leq j\leq m+1$,
\begin{align}
e_{\ell,j}(a)=&K_{m,j}(G^{s+\frac{1}{2}}(a))\omega^{(2s+1)\left(\frac{1}{2}+\frac{1-j}{m}\right)}-K_{m,j}(G^{-s-\frac{1}{2}}(a))\omega^{-(2s+1)\left(\frac{1}{2}+\frac{1-j}{m}\right)}\nonumber\\
&+K_{m,j}(G^{s-\frac{1}{2}}(a))\omega^{(2s-1)\left(\frac{1}{2}+\frac{1-j}{m}\right)}-K_{m,j}(G^{-s+\frac{1}{2}}(a))\omega^{-(2s-1)\left(\frac{1}{2}+\frac{1-j}{m}\right)}\nonumber\\
=&\sum_{k=0}^j(-1)^kK_{m,j,k}b_{j,k}(a)\left(\omega^{-j\frac{\ell+1}{2}+(\ell+1)\left(\frac{1}{2}+\frac{1-j}{m}\right)}-\omega^{j\frac{\ell+1}{2}-(\ell+1)\left(\frac{1}{2}+\frac{1-j}{m}\right)}\right.\nonumber\\
&\qquad\qquad\qquad\qquad\qquad\left.+\omega^{-j\frac{\ell-1}{2}+(\ell-1)\left(\frac{1}{2}+\frac{1-j}{m}\right)}-\omega^{j\frac{\ell-1}{2}-(\ell-1)\left(\frac{1}{2}+\frac{1-j}{m}\right)}\right)\nonumber\\
=&4i\sum_{k=0}^j(-1)^kK_{m,j,k}b_{j,k}(a)\sin\left(\frac{(1-j)\ell\pi}{m}\right)\cos\left(\frac{(1-j)\pi}{m}\right),\nonumber
\end{align}
where we used Lemma \ref{lemma_25} and $\ell=2s$. Then we use the
arguments in the case when $\ell$ is odd to show that there exists
$N_0=N_0(m)$ such that the set $\{\lambda_n\}_{n\geq N_0}$ of all
eigenvalues satisfy \eqref{asy_eq_even} and hence, the proof is
completed.
\end{proof}

\section{Proof of Theorem \ref{main_thm1} when $\ell=\lfloor\frac{m}{2}\rfloor$ and when $\frac{m}{2}<\ell\leq m-1$.}\label{sec_8}

In this section, we prove Theorem \ref{main_thm1} for
$\ell=\lfloor\frac{m}{2}\rfloor$. We first prove the theorem when
$m$ is even, and later, we will treat the cases when $m$ is odd.
Then at the end of the section, we will prove the theorem when
$\frac{m}{2}<\ell\leq m-1$, by scaling. Proof of existence of
$N_0=N_0(m)$ is the same as in Section \ref{sec_7}, so below we
will omit this part of proof.
\subsection{When $m$ is even} We further divide the case into when $\ell$ is odd and when $\ell$ is even.
\begin{proof}[Proof of Theorem ~\ref{main_thm1} when $m$ is even and $\ell=\lfloor\frac{m}{2}\rfloor$ is odd]
Suppose that $m$ is even and $\ell=\frac{m}{2}=2s-1$ for some $s\in\N$. Recall that $\lambda$ is an eigenvalue of
$H_{\ell}$ if and only if $\mathcal{W}_{-s,s}(a,\lambda)=0$ if and only if
$\mathcal{W}_{-1,2s-1}(G^{-s+1}(a),\omega^{-2s+2}\lambda)=0$.

So if $\mathcal{W}_{-1,2s-1}(G^{-s+1}(a),\omega^{-2s+2}\lambda)=0$ and if $|\lambda|$ is large enough, then from
\eqref{thm_eq1},
\begin{align}
&\exp\left[L(G^{3s-1}(a),\omega^{2s-2}\lambda)-L(G^{s+1}(a),\omega^{-2s+2}\lambda)\right]\nonumber\\
&\times\exp\left[L(G^{s-1}(a),\omega^{2s-2}\lambda)-L(G^{-s+1}(a),\omega^{-2s+2}\lambda)+O\left(\lambda^{-\rho}\right)\right]=-\omega^{4\nu(G^{s}(a))}.\nonumber
\end{align}
Then by Lemma \ref{asy_lemma},
\begin{equation}\nonumber
\left(\frac{8\nu(G^{s}(a))}{m+2}+2n+1\right)\pi
i=\sum_{j=0}^{m+1}e_{\ell,j}(a)\lambda_{n}^{\frac{1}{2}+\frac{1-j}{m}}+\frac{16(s-1)\nu(G^{s}(a))}{m(m+2)}\pi
i+O\left(\lambda^{-\rho}\right),
\end{equation}
where for $0\leq j\leq m+1$, the coefficients $e_{\ell,j}(a)$ are given by
\begin{align}
e_{\ell,j}(a)=&K_{m,j}(G^{3s-1}(a))\omega^{2(s-1)\left(\frac{1}{2}+\frac{1-j}{m}\right)}-K_{m,j}(G^{s+1}(a))\omega^{-2(s-1)\left(\frac{1}{2}+\frac{1-j}{m}\right)}\nonumber\\
&+K_{m,j}(G^{s-1}(a))\omega^{2(s-1)\left(\frac{1}{2}+\frac{1-j}{m}\right)}-K_{m,j}(G^{-s+1}(a))\omega^{-2(s-1)\left(\frac{1}{2}+\frac{1-j}{m}\right)}\nonumber\\
=&K_{m,j}(a)\left(\omega^{(s+1)j+2(s-1)\left(\frac{1}{2}+\frac{1-j}{m}\right)}-\omega^{-(s+1)j-2(s-1)\left(\frac{1}{2}+\frac{1-j}{m}\right)}\right.\nonumber\\
&\qquad\qquad\qquad\qquad\left.+\omega^{-(s-1)j+2(s-1)\left(\frac{1}{2}+\frac{1-j}{m}\right)}-\omega^{(s-1)j-2(s-1)\left(\frac{1}{2}+\frac{1-j}{m}\right)}\right)\nonumber\\
=&2iK_{m,j}(a)\left(1+(-1)^j\right)\sin\left(\frac{(1-j)(2s-2)\pi}{m}\right)\nonumber\\
=&4iK_{m,j}(a)\sin\left(\frac{(1-j)\pi}{2}\right)\cos\left(\frac{(1-j)\pi}{m}\right),\nonumber
\end{align}
where we used Lemma \ref{lemma_25}.
\end{proof}

\begin{proof}[Proof of Theorem ~\ref{main_thm1} when $m$ is even and $\ell=\lfloor\frac{m}{2}\rfloor$ is even]
Let $\ell=\frac{m}{2}=2s$ for some $s\in\N$. Recall that $\lambda$ is an eigenvalue of $H_{\ell}$ if and only if
$\mathcal{W}_{-1,2s}(G^{-s+1}(\widetilde{a}),\omega^{-2s+1}\lambda)=0$.

Suppose that
$\mathcal{W}_{-1,2s}(G^{-s+1}(\widetilde{a}),\omega^{-2s+1}\lambda)=\mathcal{W}_{-1,2s}(G^{-s+\frac{1}{2}}(a),\omega^{-2s+1}\lambda)=0$.
Then by  \eqref{thm_eq1},
\begin{align}
&\exp\left[L(G^{3s+\frac{1}{2}}(a),\omega^{2s-1}\lambda)-L((G^{s+\frac{3}{2}}(a),\omega^{-2s+1}\lambda)\right]\nonumber\\
&\times\exp\left[L(G^{s-\frac{1}{2}}(a),\omega^{2s-1}\lambda)-L(G^{-s+\frac{1}{2}}(a),\omega^{-2s+1}\lambda)+O\left(\lambda^{-\rho}\right)\right]=-1.\nonumber
\end{align}
Then like before,
\begin{equation}\nonumber
\left(2n+1\right)\pi
i=\sum_{j=0}^{m+1}e_{\ell,j}(a)\lambda_{n}^{\frac{1}{2}+\frac{1-j}{m}}+O\left(\lambda^{-\rho}\right),
\end{equation}
where for $0\leq j\leq m+1$,
\begin{align}
e_{\ell,j}(a)=&K_{m,j}(G^{3s+\frac{1}{2}}(a))\omega^{(2s-1)\left(\frac{1}{2}+\frac{1-j}{m}\right)}-K_{m,j}((G^{s+\frac{3}{2}}(a))\omega^{-(2s-1)\left(\frac{1}{2}+\frac{1-j}{m}\right)}\nonumber\\
&+K_{m,j}(G^{s-\frac{1}{2}}(a))\omega^{(2s-1)\left(\frac{1}{2}+\frac{1-j}{m}\right)}-K_{m,j}(G^{-s+\frac{1}{2}}(a))\omega^{-(2s-1)\left(\frac{1}{2}+\frac{1-j}{m}\right)}\nonumber\\
=&2i\sum_{k=0}^j(-1)^{k}K_{m,j,k}b_{j,k}(a)\left(1+(-1)^j\right)\sin\left(\frac{(1-j)(2s-1)\pi}{m}\right)\nonumber\\
=&4i\sum_{k=0}^j(-1)^kK_{m,j,k}b_{j,k}(a)\sin\left(\frac{(1-j)\pi}{2}\right)\cos\left(\frac{(1-j)\pi}{m}\right).\nonumber
\end{align}
\end{proof}

\subsection{When $m$ is odd} We divide the case into when $\ell$ is odd and when $\ell$ is even.
\begin{proof}[Proof of Theorem ~\ref{main_thm1} when $m$ and $\ell=\lfloor\frac{m}{2}\rfloor$ are odd]
Let $m$ and $\ell=\frac{m-1}{2}=2s-1$ be odd. Suppose that $\lambda$ is an eigenvalue of $H_{\ell}$. Since
$\lambda$ is an eigenvalue of $H_{\ell}$ if and only if
$\mathcal{W}_{-1,2s-1}(G^{-s+1}(a),\omega^{-2s+2}\lambda)=0$,
 if  $|\lambda|$ is large enough, then from  \eqref{thm_eq2},
\begin{align}
&\exp\left[L(G^{s-1}(a),\omega^{2s-2}\lambda)-L(G^{-s+1}(a),\omega^{-2s+2}\lambda)\right]\nonumber\\
&\times\exp\left[-L(G^{3s+1}(a),\omega^{-2s}\lambda)-L(G^{s+1}(a),\omega^{-2s+1}\lambda)+O\left(\lambda^{-\rho}\right)\right]=-1.\nonumber
\end{align}
Then, from Lemma \ref{asy_lemma},
\begin{equation}\nonumber
\left(2n+1\right)\pi
i=\sum_{j=0}^{m+1}e_{\ell,j}(a)\lambda_{n}^{\frac{1}{2}+\frac{1-j}{m}}+O\left(\lambda^{-\rho}\right),
\end{equation}
where for $0\leq j\leq m+1$,
\begin{align}
e_{\ell,j}(a)=&K_{m,j}(G^{s-1}(a))\omega^{(2s-2)\left(\frac{1}{2}+\frac{1-j}{m}\right)}-K_{m,j}(G^{-s+1}(a))\omega^{(-2s+2)\left(\frac{1}{2}+\frac{1-j}{m}\right)}\nonumber\\
&-K_{m,j}(G^{3s+1}(a))\omega^{-2s\left(\frac{1}{2}+\frac{1-j}{m}\right)}-K_{m,j}(G^{s+1}(a))\omega^{(-2s+1)\left(\frac{1}{2}+\frac{1-j}{m}\right)}\nonumber\\
=&K_{m,j}(a)\left(\omega^{-j(s-1)+(2s-2)\left(\frac{1}{2}+\frac{1-j}{m}\right)}-\omega^{j(s-1)-(2s-2)\left(\frac{1}{2}+\frac{1-j}{m}\right)}\right.\nonumber\\
&\qquad\qquad\qquad\left.-\omega^{-j(3s+1)-2s\left(\frac{1}{2}+\frac{1-j}{m}\right)}-\omega^{-j(s+1)-(2s-1)\left(\frac{1}{2}+\frac{1-j}{m}\right)}\right)\nonumber\\
=&K_{m,j}(a)\left(\omega^{(1-j)(2s-2)\left(\frac{1}{2}+\frac{1}{m}\right)}-\omega^{-(1-j)(2s-2)\left(\frac{1}{2}+\frac{1}{m}\right)}\right.\nonumber\\
&\qquad\qquad\qquad\left.-\omega^{js-2s\left(\frac{1}{2}+\frac{1-j}{m}\right)}+\omega^{-js+2s\left(\frac{1}{2}+\frac{1-j}{m}\right)}\right)\nonumber\\
=&2iK_{m,j}(a)\left(\sin\left(\frac{(m-3)(1-j)\pi}{2m}\right)+\sin\left(\frac{(m+1)(1-j)\pi}{2m}\right)\right)\nonumber\\
=&4iK_{m,j}(a)\sin\left(\frac{(1-j)\ell\pi}{m}\right)\cos\left(\frac{(1-j)\pi}{m}\right).\nonumber
\end{align}
\end{proof}

\begin{proof}[Proof of Theorem ~\ref{main_thm1} when $m$ is odd and $\ell=\lfloor\frac{m}{2}\rfloor$ is even]
Let $\ell=\frac{m-1}{2}=2s$ for some $s\in\N$. Then $\lambda$ is an eigenvalue of $H_{\ell}$ if and only if
$\mathcal{W}_{-1,2s}(G^{-s+1}(\widetilde{a}),\omega^{-2s+1}\lambda)=0$.

If $\mathcal{W}_{-1,2s}(G^{-s+1}(\widetilde{a}),\omega^{-2s+1}\lambda)=0$ and if $|\lambda|$ is large enough, then
from \eqref{thm_eq2},
\begin{align}
&\exp\left[L(G^{s}(\widetilde{a}),\omega^{2s-1}\lambda)-L(G^{-s+1}(\widetilde{a}),\omega^{-2s+1}\lambda)\right]\nonumber\\
&\times\exp\left[-L(G^{3s+3}(\widetilde{a}),\omega^{-2s-1}\lambda)-L(G^{s+2}(\widetilde{a}),\omega^{-2s}\lambda)+O\left(\lambda^{-\rho}\right)\right]=-1.\nonumber
\end{align}

Then,
from Lemma \ref{asy_lemma},
\begin{equation}\nonumber
\left(2n+1\right)\pi
i=\sum_{j=0}^{m+1}e_{\ell,j}(a)\lambda_{n}^{\frac{1}{2}+\frac{1-j}{m}}+O\left(\lambda^{-\rho}\right),
\end{equation}
where for $0\leq j\leq m+1$,
\begin{align}
e_{\ell,j}(a)=&K_{m,j}(G^{s-\frac{1}{2}}(a))\omega^{(2s-1)\left(\frac{1}{2}+\frac{1-j}{m}\right)}-K_{m,j}((G^{-s+\frac{1}{2}}(a))\omega^{-(2s-1)\left(\frac{1}{2}+\frac{1-j}{m}\right)}\nonumber\\
&-K_{m,j}(G^{3s+\frac{5}{2}}(a))\omega^{-(2s+1)\left(\frac{1}{2}+\frac{1-j}{m}\right)}-K_{m,j}(G^{s+\frac{3}{2}}(a))\omega^{-2s\left(\frac{1}{2}+\frac{1-j}{m}\right)}\nonumber\\
=&2i\sum_{k=0}^j(-1)^kK_{m,j,k}b_{j,k}(a)\left[\sin\left(\frac{(1-j)(2s-1)\pi}{m}\right)+\sin\left(\frac{(1-j)(2s+1)\pi}{m}\right)\right]\nonumber\\
=&4i\sum_{k=0}^j(-1)^kK_{m,j,k}b_{j,k}(a)\sin\left(\frac{(1-j)\ell\pi}{m}\right)\cos\left(\frac{(1-j)\pi}{m}\right).\nonumber
\end{align}
\end{proof}
Theorem \ref{main_thm1} for $1\leq\ell\leq\frac{m}{2}$ has been proved. Next we prove Theorem \ref{main_thm1} for
$\frac{m}{2}<\ell\leq m-1$, by the change of the variables $z\mapsto -z$.
\subsection{When $\frac{m}{2}<\ell\leq m-1$}
\begin{proof}[Proof of Theorem ~\ref{main_thm1} when $\frac{m}{2}<\ell\leq m-1$]
Suppose that  $\ell>\frac{m}{2}$. If $u$ is an eigenfunction of $H_{\ell}$, then $v(z)=u(-z,\lambda)$ solves
\begin{equation}\nonumber
-v^\dd(z)+\left[(-1)^{-\ell}(-iz)^m-P(-iz)\right]v(z) =\lambda v(z),
\end{equation}
and
\begin{equation}\nonumber
\text{$v(z)\rightarrow 0$ exponentially, as $z\rightarrow \infty$ along the two rays}\quad \arg z=-\frac{\pi}{2}\pm \frac{((m-\ell)+1)\pi}{m+2}.
\end{equation}
The coefficient vector of $P(-z)$ is $((-1)^{m-1}a_1,(-1)^{m-2}a_2,\dots,(-1)^{1}a_{m-1}, a_m).$ Certainly,
$$\sin\left(\frac{(1-j)(m-\ell)\pi}{m}\right)=(-1)^j\sin\left(\frac{(1-j)\ell\pi}{m}\right).$$
Also, one can find from Lemma \ref{lemma_25} that for $0\leq k\leq j$,
$$b_{j,k}((-1)^{m-1}a_1,(-1)^{m-2}a_2,\dots,-a_{m-1}, a_m)=(-1)^{mk-j}b_{j,k}(a_1,a_2,\dots,a_{m-1},a_m).$$
Moreover, $c_{m-\ell,j}((-1)^{m-1}a_1,(-1)^{m-2}a_2,\dots,-a_{m-1}, a_m)=c_{\ell,j}(a_1,a_2,\dots,a_{m-1}, a_m).$
This completes proof of Theorem \ref{main_thm1}.
\end{proof}

\subsection*{{\bf Acknowledgments}}
The author thanks Mark Ashbaugh, Fritz Gesztesy, Richard Laugesen, Boris Mityagin, Grigori Rozenblioum, and Alexander Turbiner for helpful discussions and references.

{\sc email contact:}  kshin@westga.edu
\end{document}